%% file: arxiv.tex
\documentclass[12pt]{article}
\usepackage{fullpage}
\usepackage{amsmath} 
\usepackage{amsthm}
\usepackage{amssymb}
\usepackage[mathscr]{euscript}
\usepackage{prooftree}
\usepackage{multirow}
\usepackage[inline]{enumitem}
\usepackage{mathtools}
\usepackage{thmtools} 
\usepackage{thm-restate}
\usepackage{tabularx}
\usepackage{multirow}
\usepackage{subcaption}
\usepackage{algorithm}
\usepackage{algpseudocode}
\usepackage{tabu}
\usepackage[compress,sort]{cite}
\usepackage[colorlinks,citecolor=blue]{hyperref}

\usepackage{graphicx}
\usepackage{tikz}
\usetikzlibrary{arrows,shapes.multipart, matrix, positioning}
\usetikzlibrary{shapes.geometric,calc}

\input{macros}

\newtheorem{theorem}{Theorem}
\newtheorem{example}[theorem]{Example}

\newtheorem{obs}[theorem]{Observation}

\newtheorem{lemma}[theorem]{Lemma}
\newtheorem{definition}[theorem]{Definition}

\title{Solving the insecurity problem for assertions}
\author{
    R Ramanujam\\
    \textit{The Institute of Mathematical Sciences, Chennai (Retd.)} \\
    \textit{Homi Bhabha National Institute, Mumbai (Retd.)} \\
    \textit{Azim Premji University, Bengaluru (Visiting)} \\
    Bengaluru, India \\
    Email: \texttt{jam@imsc.res.in}
    \and
    Vaishnavi Sundararajan \\
    \textit{Dept of Computer Science \& Engineering} \\
    \textit{Indian Institute of Technology Delhi}\\
    New Delhi, India \\
    Email: \texttt{vaishnavi@cse.iitd.ac.in}
    \and
    S P Suresh\thanks{Partially supported by a grant from the Infosys Foundation.}\\
    \textit{Chennai Mathematical Institute} \\
    \textit{CNRS UMI 2000 ReLaX} \\
    Chennai, India \\
    Email: \texttt{spsuresh@cmi.ac.in}
}
\date{}

\begin{document}
\maketitle

\begin{abstract}
    In the symbolic verification of cryptographic protocols, a central problem is deciding whether a protocol admits an execution which leaks a designated secret to the malicious intruder. In~\cite{RT03}, it is shown that, when considering finitely many sessions, this ``insecurity problem'' is NP-complete. Central to their proof strategy is the observation that any execution of a protocol can be simulated by one where the intruder only communicates terms of bounded size.  However, when we consider models where, in addition to terms, one can also communicate logical statements about terms, the analysis of the insecurity problem becomes tricky when both these inference systems are considered together. In this paper we consider the insecurity problem for protocols with logical statements that include {\em equality on terms} and {\em existential quantification}. Witnesses for existential quantifiers may be unbounded, and obtaining small witness terms while maintaining equality proofs complicates the analysis considerably. We extend techniques from~\cite{RT03} to show that this problem is also in NP.
\end{abstract}

\section{Introduction}
\subsection{Symbolic analysis of cryptographic protocols}
Symbolic analysis of security protocols is a long-standing field of study, with the Dolev-Yao model~\cite{DY83} being the standard. In this model, cryptographic operations are abstracted as operators in a term algebra, and the ability to build new messages from old ones is specified by rewrite rules or a proof system. The model includes an intruder who controls the network, and can see, block, inject, redirect, as well as derive terms, but cannot break cryptography. Informally, protocols are specified as a finite sequence of \emph{communications} between \emph{principals/agents}. We now illustrate this model using an example.

\begin{example}\label{ex:exprot}
Alice sends to Bob her public key as well as a randomly-chosen value encrypted in Bob's public key. Bob receives it, decrypts it using his private key, encrypts it in Alice's public key, and sends it back to her.  We split each communication into a send and a receive. We formalize the protocol as two \emph{roles}: an \emph{initiator} role ${{\sf init}}(A,B)$ (left column) and a \emph{responder} role ${{\sf resp}}(B)$ (right column). We use $!C$ and $?C$ to denote a send and a receive respectively by an agent $C \in \{A, B\}$. $k_{A}$ and $k_{B}$ stand for the private keys of $A$ and $B$ respectively, $\pk(k)$ stands for the public key corresponding to a key $k$, and $\{t\}_{k}$ stands for the encryption of a message $t$ using a key $k$. 

\begin{minipage}[t]{0.45\columnwidth}
    \vspace{-3mm}
    \begin{align*}
        A &: \text{Generate fresh}\ m \\
        !A &: (\pk(k_{A}), \{m\}_{\pk(k_{B})}) \\
        ?A &: \{m\}_{\pk(k_{A})}
    \end{align*}
\end{minipage}
\vspace{5mm}
\hfill{\vline width 2pt}\hfill
\begin{minipage}[t]{0.4\columnwidth}
    \vspace{-3mm}
    \begin{align*}
        ?B &: (x, \{y\}_{\pk(k_{B})}) \\
        !B &: \{y\}_{x} 
    \end{align*}
\end{minipage}

The protocol itself can be thought of as a program running potentially unboundedly-many copies (\emph{sessions}) of ${{\sf init}}$ and ${{\sf resp}}$ in parallel. Each copy instantiates parameters $A$ and $B$ with agent names, while  $x$ and $y$ denote parts of messages received while participating in a session, and will be instantiated accordingly.  An \emph{execution} (\emph{run}) of a protocol is an interleaving of a finite set of sessions, such that every sent message can be generated by the sender (based on their current knowledge), and received messages by the intruder $I$ (since every received message comes from the channel, and could have been potentially tampered with by the intruder). 
\end{example}

Is there any execution of this protocol at the end of which the intruder can derive $m$? This property is called \emph{confidentiality}. In fact, the intruder can effect the following man-in-the-middle attack, at the end of which $A$ thinks $m$ is secret between her and $B$, while $B$ thinks $m$ is secret between him and $I$. $B$ receives a message where $x$ can be matched with $\pk(k_{I})$ and $y$ with $m$, and thus sends out $\{m\}_{\pk(k_{I})}$.
\begin{align*}
    !A &: (\pk(k_{A}), \{m\}_{\pk(k_{B})}) & & \\
    & & ?B &: (\pk(k_{I}), \{m\}_{\pk(k_{B})}) \\
    & & !B &: \{m\}_{\pk(k_{I})} \\
    ?A &: \{m\}_{\pk(k_{A})}
\end{align*}

\subsection{Communicating ``assertions''}
The Dolev-Yao model and its extensions have been studied extensively over the last forty years. People have studied extensions that express richer classes of protocols and security properties~\cite{ABF17, Bla01, BMU08, CDL06}, and associated decidability and complexity results~\cite{BP05,BRS10,CS03,CKRT05,LLT07,CDS21,RS05,RS06,DLMS04,Bau05, CKR20, AC06, CRZ05}. Various verification tools have also been built based on these formal models \cite{Cre08,Bla01,Bla16,MSCB13,CKR18}.

In this paper, we consider an extension introduced in~\cite{RSS17}, which gives agents the power to communicate terms as well as logical formulas about them. These formulas, called \emph{assertions}, involve equality of terms, existential quantification, conjunction, and disjunction. For instance, we can reveal partial information about some encrypted term $\{m\}_{k}$ to a recipient who does not know the key $k$ (for instance, that the value of $m$ is either $0$ or $1$, without revealing which) by sending the assertion $\exists{x}{y}\bigl[\{x\}_{y} = \{m\}_{k} \ \conj\ x \in \{0,1\}\bigr]$. So we see that assertions allow us to model protocols that involve some kinds of \emph{certification}. Traditionally, such certification is often modelled using \emph{zero-knowledge proofs}. 

The Dolev-Yao model can also be extended with a special class of zero-knowledge terms~\cite{BMU08, BHM08b}. But in these extensions, one important component is missing: logical reasoning over certificates. This is especially important in situations where certificates communicate partial information. For example, two partial-information certificates of the form $x \in \{0, 1\}$  and $x \in \{0, 2\}$ can lead to the inference of strictly greater information, namely $x = 0$, potentially violating some security guarantees. This is one of the main features of the model in~\cite{RSS17}. Making ``assertions'', as that paper refers to such logical statements, first-class citizens provides a threefold advantage: a more transparent specification of protocols which captures design intent better, the ability to explicitly reason about certificates and thus analyze protocols more precisely, and the ability to state some security properties more easily. In~\cite{RSS17}, the authors express examples (the FOO~\cite{FOO92} and Helios~\cite{Adi08} e-voting protocols) and specify security properties using assertions. We describe the modelling of the FOO protocol in detail in Section~\ref{sec:foo}.

In~\cite{RSS17}, any communicated assertion is ``believed'' by the recipients. One way to implement this feature is to communicate a zero knowledge proof of the assertion. But formally, we send the assertion itself rather than a term standing for a zero-knowledge proof, which also allows us the possibility of choosing other implementations for the assertion. Another way in which~\cite{RSS17} differs from other modelling using ZKP terms is that these proofs need not be built ab initio every time. One can compose a new proof by combining existing proofs. These can be implemented using composable ZKPs~\cite{GS08}. These issues have been discussed in~\cite{MPR13}, which considers a logical language with conjunction and existential quantification and modular construction of ZKPs for these formulas. However, unlike~\cite{MPR13}, assertions also allow ``destructive reasoning'' from existing knowledge via elimination rules.

The main focus in this paper is to solve an interesting technical problem in our model with assertions -- the \emph{insecurity problem for finitely many sessions}. 

\subsection{The insecurity problem for finitely many sessions}
The attack on Example~\ref{ex:exprot} indicates that even for simple protocols, one needs to consider non-trivial scenarios to detect security violations. A canonical problem of interest is the \emph{insecurity problem}, which asks if a given protocol admits a run that leaks a secret to the intruder. A run is characterized by an interleaving of protocol roles ($A$ and $B$ in Example~\ref{ex:exprot}), with a substitution for the variables in messages received by agents during these roles. There can be infinitely many such substitutions, i.e. a potentially infinite number of executions, and thus, the insecurity problem is undecidable in general~\cite{ALV03, DLMS04, HT96}. In \cite{RT03}, the authors consider a restricted set of runs, and show that the insecurity problem is in NP when one considers at most $K$ sessions, for some fixed $K$. 

Even with only a finite number of sessions, the intruder can inject arbitrarily large terms in place of variables. Thus, there is no bound on the size of terms encountered in a run. The work in~\cite{RT03} gets around this complication by showing that if there is any attack at all given by an interleaving of roles and a substitution, there is an attack given by the same interleaving and a `small' substitution. This ``new'' attack is such that the intruder can derive the same terms at the end, and the size of all messages transmitted is bounded by a polynomial in the size of the protocol specification. Hence the insecurity problem with boundedly many sessions can be solved in NP. 

As with terms, one can formulate the insecurity problem for assertions as well. The general problem continues to be undecidable, so we consider the case of finitely many sessions. With existential quantification, we now have two types of variables -- those used to identify parts of received messages (instantiated at runtime by the actual message sent by the intruder), and quantified variables that occur in assertions. As earlier, there is no a priori bound on the size of terms assigned to the first kind of variables. But there is another source of unboundedness: to derive a quantified assertion $\exists{x}.~\alpha$, one must derive $\alpha(t)$ for some ``witness'' $t$. There is no a priori bound on the size of $t$ either, and proof search is further complicated by any potential interaction between these two sources of unboundedness. When we simulate a substitution for the ``intruder'' variables with a small one, the witnesses for quantifiers might change too, but we still need to preserve some derivations under these new witnesses. 

We extend the techniques of~\cite{RT03}, while considering interactions between multiple substitutions and having to preserve more complex derivations, to obtain a somewhat surprising result. In this paper, we show that the insecurity problem for assertions for finitely many sessions remains in NP. 

\subsection{Related work}
\label{sec:related}
There are many extensions of the basic Dolev-Yao model that aim to capture various cryptographic operators and their properties~\cite{AC06,BRS10,CKRT05,CS03,CDL06,CRZ05,LLT07}. Algebraic properties of operators like \textsc{xor}, blinding, distributive encryption etc. are studied by means of \emph{equation theories}, which are also referred to as \emph{intruder theories} in the security literature. Equations in these theories are implicitly universally quantified, and the intention is that any term matching one side of the equation may be replaced by the other side. For example, if the theory contains a rule of the form $\textit{unblind}(\textit{sign}(\textit{blind}(x,y),k),y) = \textit{sign}(x,k)$, it means that any instance of the LHS can be replaced by the corresponding instance of the RHS. Such equations correspond to proof rules in the system for deriving terms in this paper (examples of such systems are given in Section~\ref{subsec:terms}).

Equality assertions, on the other hand, are to be treated literally, and not as rewrite rules. For instance, given an assertion of the form $\{x\}_{k} = \{t\}_{k}$, we cannot replace all terms of the form $\{u\}_{k}$ by $\{t\}_{k}$. In fact, these equality assertions are objects that are manipulated by proof rules, rather than being another style of expressing derivations between terms. 

Along with studying the derivability problem for such extensions, several of these papers also extend the results of~\cite{RT03} by addressing the active intruder problem for finitely many sessions. For instance, \cite{CKRT05,CS03} obtain NP decision procedures in the case of extending Dolev-Yao with rules for \textsc{xor}. The current paper, however, extends~\cite{RT03} along a different dimension, to solve both the passive and active intruder problems for assertions, and is thus not subsumed by any of these works on equation theories. 

\subsection{Organization of the paper}
In Section~\ref{sec:prot}, we first introduce the syntax for terms and assertions. We present an example of modelling with assertions via the FOO e-voting protocol, and then present the proof system for assertions. Then we define protocols and runs for this new system. In Section~\ref{sec:proofstrat}, we first present a high-level overview of the various steps involved in solving the insecurity problem, and then we move on to Section~\ref{sec:insecurity}, where we present the technical results in detail and prove that insecurity for the assertion system is in NP. We present some ideas for future research in Section~\ref{sec:disc}.

\section{Modeling security protocols}\label{sec:prot}
\subsection{Terms: Syntax and Derivation System}\label{subsec:terms}
In this model, each communicated message is modelled as a term in an algebra, which has operators for pairing, encryption, hashing etc. New terms can be derived from old ones using proof rules, which specify the behaviour of these operators. We begin with a set $\names$ of names (atomic terms, with no further structure), and a set of variables $\vars$. We denote by $\ag \subseteq \names$ the set of agents, with $I \in \ag$ being the malicious intruder. We denote by $\qvar \subset \vars$ the variables used for quantification, and by $\ivar$ the set $\vars\setminus\qvar$. The set of terms, denoted by $\Tterms$, is given by
\[
    t \in \Tterms ::= x \mid m \mid \func (t_{1}, \ldots, t_{n})
\]
where $x \in \vars$,  $m \in \names$, $t_{1}, \ldots t_{n} \in \Tterms$, and $\func$ is an $n$-ary operator. The set of \emph{ground} terms are those without variables. A substitution $\sigma$ is a partial function with finite support from $\ivar$ to $\Tterms$. Its domain is denoted by $\dom(\sigma)$. We assume that $\sigma(x) = x$ for $x \not\in \dom(\sigma)$. The set of subterms of $t$ is denoted by $\subterms(t)$, and defined as usual. The set of variables appearing in $t$ is denoted by $\varsof(t)$.

Each $\func$ has constructor rules and destructor rules, expressed in terms of sequents of the form $X \vdash t$ (to be read as ``$t$ is derived from $X$''), where $X \cup \{t\}$ is a finite set of terms. Figure~\ref{fig:consdest} gives the general form of a constructor rule (on the left) and a destructor rule (on the right). In a destructor rule, the conclusion $t_{i}$ is an immediate subterm of the leftmost premise, which is designated as the \emph{major premise} of the rule. The $\rnax$ rule (which derives $X \vdash t$ when $t \in X$) is also considered a destructor rule for technical purposes. We say $X \DYderives t$ if there is a proof of $X \vdash t$ using these constructor and destructor rules, and $X \DYderives S$ to mean that $X \DYderives t$ for every $t \in S$.

\begin{figure}
    \centering
    \small
    \begin{tabu}{cc}
        \begin{prooftree}
            X \vdash t_{1} \cdots X \vdash t_{n}
            \justifies X \vdash \func(t_{1},\ldots,t_{n})
        \end{prooftree}
        &
        \begin{prooftree}
            X \vdash \func(t_{1},\ldots,t_{n}) \quad X \vdash u_{1}  \cdots X \vdash u_{m}
            \justifies X \vdash t_{i}
        \end{prooftree}
    \end{tabu}
    \caption{General form of constructor and destructor rules}
    \label{fig:consdest}
\end{figure}

For any proof $\pi$ of $X \vdash t$, we denote by $\axiomsof(\pi)$ the set $X$, by $\concof(\pi)$ the term $t$, and by $\termsof(\pi)$ all terms occurring in $\pi$. 
$\pi$ is said to be normal if a constructor rule does not yield the major premise of a destructor rule. We only consider proof systems which enjoy the following three properties:
\begin{itemize}
    \item \emph{Normalization}: Every proof $\pi$ of $X \vdash t$ can be converted into a normal proof $\varpi$ of the same.
    \item \emph{Subterm property}: For any normal proof $\varpi$ of $X \vdash t$, $\termsof(\varpi) \subseteq \subterms(X \cup \{t\})$, and if $\varpi$ ends in a destructor rule, $\termsof(\varpi) \subseteq \subterms(X)$.
    \item \emph{Efficient derivability checks:} There is a PTIME algorithm for checking derivability.
\end{itemize}
The normalization and subterm properties combined are referred to as \emph{locality} in the security literature. This is a notion identified in~\cite{McAll93}, and is crucially used in solving the derivability problem for many classes of inference systems, including many intruder theories.  

\begin{example}\label{ex:termalg}
    A term algebra with pairing, symmetric and asymmetric encryption operations, where $m, k \in \names$ and $t, u \in \Tterms$ is given by
$t := m \mid \pk(k) \mid (t, u) \mid \{t\}_{k} \mid \{\!|t|\!\}_{\pk(k)}$. The proof system for this algebra is shown in Table~\ref{tab:termalgtab}. This system enjoys normalization and the subterm property~\cite{RT03}.
\end{example}

\begin{table}
    \centering
    \tabulinesep=0.8mm
    \setlength{\tabcolsep}{0.4em}
    \begin{tabu}{|c|c|c|c|c|c|}
        \hline
        \multicolumn{2}{|c|}{
            \begin{prooftree}
                \justifies X \vdash m \using \rnax (m \in X)
            \end{prooftree}
        }
        &
        \multicolumn{2}{c|}{
            \begin{prooftree}
                X \vdash (t_1, t_2)
                \justifies X \vdash t_i \using \rnsplit
            \end{prooftree}
        }
        &
        \multicolumn{2}{c|}{
            \begin{prooftree}
                X \vdash k 
                \justifies X \vdash \pk(k) \using \rnpk
            \end{prooftree}
        }
        \\
        \hline
        \multicolumn{2}{|c|}{
            \begin{prooftree}
                X \vdash t \quad X \vdash u
                \justifies X \vdash (t, u) \using \rnpair
            \end{prooftree}
        }
        &
        \multicolumn{2}{c|}{
            \begin{prooftree}
                X \vdash \{t\}_{k} \quad X \vdash k
                \justifies X \vdash t \using \rnsdec
            \end{prooftree}
        }
        &
        \multicolumn{2}{c|}{
            \begin{prooftree}
                X \vdash t \quad X \vdash k
                \justifies X \vdash \{t\}_{k} \using \rnsenc
            \end{prooftree}
        }
        \\
        \hline
        \multicolumn{3}{|c|}{
            \begin{prooftree}
                X \vdash \{\!|t|\!\}_{\pk(k)} \quad X \vdash k
                \justifies X \vdash t \using \rnadec
            \end{prooftree}
        }
        &
        \multicolumn{3}{c|}{
            \begin{prooftree}
                X \vdash t \quad X \vdash \pk(k)
                \justifies X \vdash \{\!|t|\!\}_{\pk(k)} \using \rnaenc
            \end{prooftree}
        }
        \\
        \hline
    \end{tabu}
    \caption{Proof system for the term algebra in Example~\ref{ex:termalg}}
    \label{tab:termalgtab}
\end{table}

\subsection{Assertions}
We consider an assertion syntax which includes equality over terms (to avoid overloading the $=$ operator, we denote equality between $t$ and $u$ by $\equals{t}{u}$), predicates, conjunction, existentially quantified assertions, list membership, and a$\says$connective. Existential quantification allows us to make statements that convey partial information about terms, in particular, allowing us to hide terms or parts thereof. The$\says$connective works like a signature over assertions, indicating who endorses the fact conveyed by the assertion. List membership, which we denote by $\listmemb$, acts as a restricted form of disjunction. Predicates allow us to express some protocol-specific facts. As we will see over the later sections, this fragment allows us to express example protocols of interest, as well as yields a decidable active intruder problem for boundedly many sessions.

In the following, $t,u \in \Tterms$, $P$ is an $m$-ary predicate, $u_{1}, \ldots, u_{m}, t_{0} \in \names \cup \vars$, and $t_{1}, \ldots, t_{n} \in \names$,\footnote{We could consider arbitrary terms in list membership, but this simple syntax suffices for most examples. Similarly for $P(u_{1}, \ldots. u_{m})$.} $x \in \qvar$, and $\pk(k)$ is the public key corresponding to a secret key $k$.
\begin{align*}
    \alpha &:= \equals{t}{u} \mid P(u_{1}, \ldots, u_{m}) \mid t_{0} \listmemb [t_{1}, \ldots, t_{n}] \\
    & \hspace{5mm} \mid \alpha_{0} \conj \alpha_{1} \mid \exists x.~\alpha(x) \mid \pk(k) \says \alpha
\end{align*}

By \emph{atomic assertions}, we mean assertions that are not of the form $\alpha\conj\beta$ or $\exists{x}\alpha$. 

We denote the free (resp. bound) variables occurring in an assertion $\alpha$ by $\freevars(\alpha)$ and $\boundvars(\alpha)$. $\varsof(\alpha) = \freevars(\alpha) \cup \boundvars(\alpha)$. The set of subterms (resp. subformulas) of $\alpha$ is given by $\subterms(\alpha)$ (resp. $\SF(\alpha)$). We can lift these notions to sets of assertions as usual. For a substitution $\lambda$, we obtain $\lambda(\alpha)$ by replacing $x$ in $\alpha$ by $\lambda(x)$ for all $x \in \freevars(\alpha)$. 

We now define the \emph{public terms} of an assertion $\alpha$. These are essentially the terms that $\alpha$ is ``about'', which are always communicated along with $\alpha$. Quantified variables in an assertion stand for ``private'' terms, so if a term $t$ occurring in $\alpha$ has quantified variables, it cannot itself be public. But it is not reasonable to declare all other subterms to be public terms either. For instance, if an assertion talks about $\senc(v,k)$, the term $\senc(v,k)$ should be public, but probably not $v$ or $k$ itself. Hence we define the {public terms} of $\alpha$, denoted $\publics(\alpha)$, as the set of all \emph{maximal} subterms of $\alpha$ which contain no quantified variables. In other words, $t \in \publics(\alpha)$ iff $t \in \subterms(\alpha)$, $\varsof(t) \cap \qvar = \emptyset$, and $\forall u \in \subterms(\alpha): \ t \in \subterms(u) \implies \varsof(u) \cap \qvar \neq \emptyset$. 

\begin{example}
$A$ (with secret key $k$) encrypts a vote $v$ in a key $r$ unknown to $B$ and states that it is one of two allowed values.
\[ A \rightarrow B : \{v\}_{r}, \  \pk(k) \says\bigl\{\exists xy.\equals{\{x\}_{y}}{\{v\}_{r}} \conj x \listmemb [0, 1] \bigr\} \]
The set of public terms of this assertion is $\bigl\{ \{v\}_{r}, 0, 1 \bigr\}$.
\end{example}

Assertions, like terms, can be involved in sends and receives. However, since assertions are logical formulas, we can also have agents check them for derivability and take some action based on the result of this check, without  any send/receive. We call such an action an $\assertact$. As part of an $\assertact~{\alpha}$ action, an agent $A$ checks to see if $\alpha$ is derivable from their current knowledge. If it is, $A$ continues with their role, otherwise $A$ aborts. An $\assertact$ action allows us to model some minimal branching based on the derivability of assertions from agents' local states.

Note that this does not involve any absolute notion of the ``truth'' (or lack thereof) of an assertion. An agent can only locally check if an assertion can be ``verified'', i.e. obtained from what they know about the system at that point in the execution. It might well be the case that while an $\assertact~\alpha$ check passes for an agent $A$, a different agent $B$ might not have enough information to be able to derive $\alpha$, and abort. Conversely, if some agent's internal state has been compromised somehow and made inconsistent, they might even be able to $\assertact$ something like $0 = 1$, which is patently false. We are only concerned with the verifiability of assertions, and not their absolute truth values. 

Having introduced this system, we now present the modelling of the well-known FOO e-voting protocol~\cite{FOO92}. This is a minor modification of the presentation in~\cite{RSS17}.

\subsection{Example: FOO e-voting Protocol}\label{sec:foo}
The FOO e-voting protocol was proposed in 1992 and closely mirrors the way one votes offline. There is a voter $V$, an authority $A$ who verifies voter identities, and a collector $C$ who computes the final tally. 

To model this using only terms~\cite{FOO92, KR05}, \emph{blinding} is used. One can use $t$ and $b$ to make a \emph{blind pair} $\blind(t, b)$, and get $\sign(t, k)$ from $\sign(\blind(t, b), k) $ and $b$. The voter authenticates themselves to the authority using their signing key $\sk_{V}$, and uses the blinding operation to have the authority certify it without knowing the actual vote. The authority's signature $\sign(\cdot, \sk_{A})$ percolates through to the vote when the voter removes the blind, and the voter can then anonymously send (denoted by $\looparrowright$) this signed vote to the collector for inclusion into the final tally. This specification is shown below.
\begin{align*}
    V \rightarrow A &: \sign(\blind(\{v\}_{r}, b), \sk_{V}) \\
    A \rightarrow V &: \sign(\blind(\{v\}_{r}, b), \sk_{A}) \\
    V \looparrowright C &: \sign(\{v\}_{r}, \sk_{A}) 
\end{align*}

We model the voting phase of FOO as below, following~\cite{RSS17}. We use $\{\alpha\}^{A}$ as shorthand for $A \says \alpha$. In fact, the use of assertions allows one to also specify an eligibility check for voters via an $\assertact$. If the user is not eligible, the protocol aborts. Further, voters can also state that their vote is for an allowable candidate from the list $\ell$. These are left implicit in the terms-only modelling. 

\begin{center}
{\small
\vspace{-1.8em}
\begin{align*}
        V \rightarrow A &: \{v\}_{p}, \bigl\{\exists xr.\equals{\{x\}_{r}}{\{v\}_{p}} \conj x \listmemb \ell \bigr\}^{V} \\
        A &: \assertact~\elg(V) \\
        A \rightarrow V &: \bigl\{ \elg(V) \conj \bigl\{\exists xr.\equals{\{x\}_{r}}{\{v\}_{p}} \conj x \listmemb \ell \bigr\}^{V} \bigr\}^{A} \\
        V \looparrowright C &: \{v\}_{q}, \exists Uys.\bigl\{ \elg(U) \conj \bigl\{\exists xr.\equals{\{x\}_{r}}{\{y\}_{s}} \conj x \listmemb \ell \bigr\}^{U} \bigr\}^{A} \\
        & \hspace{5mm} \conj \bigl\{\exists w.\equals{\{y\}_{w}}{\{v\}_{q}}\bigr\}
\end{align*}
}
\end{center}

$V$ first sends to $A$ their encrypted vote along with an assertion claiming that it is for a candidate from the list $\ell$. The authority checks the voter's eligibility via the $\assertact$ action on the $\elg$ predicate. If the check passes, the authority issues a certificate stating that the voter is allowed to vote, crucially, without modifying the term containing the vote. $V$ then existentially quantifies out their name from this certificate, and anonymously sends to $C$ a re-encryption of the vote authorized by $A$ along with a certificate to that effect. Here, $p$ and $q$ are freshly-generated ephemeral keys. Thus, the intent behind the various communications is made more transparent than in the model with blind signatures. One can show that this satisfies anonymity~\cite{RSS17}.

One can also specify security properties in a more natural manner (as compared to the terms-only model). For instance, one can say that \emph{vote secrecy} is ensured in the above protocol if there is no run where the intruder can derive the assertion $\exists{xy}:[\{v\}_{p} = \{x\}_{y} \wedge x = v]$. Note that this means that while anyone can derive the value of $v$, which is public, they should not be able to identify the value inside the encrypted vote $\{v\}_{p}$ as being a particular public name. To express this in the terms-only formulation, one has to check whether two runs that only differ in the vote $v$ can be distinguished by the intruder~\cite{CK14}. It can be seen from~\cite{RSS17} that proving such properties might involve considering multiple runs simultaneously, but their specification itself does not refer to a notion of equivalence. 

\begin{example}\label{ex:disje}
Consider a protocol where $V$ sends to $A$ the vote encrypted in a fresh key $k$, and an assertion that the vote belongs to an allowable list $\ell$ of candidates. This looks as follows. $V \rightarrow A : \{v\}_{k}, \exists xr.\bigl\{ \equals{\{x\}_{r}}{\{v\}_{k}} \conj x \listmemb \ell \bigr\}$.

Suppose this same protocol is used for two elections that $V$ participates in simultaneously, where the first election has candidates $0$ and $1$ (so $\ell_{1} = [0, 1]$) and the second has candidates $0$ and $2$ (so $\ell_{2} = [0, 2]$). 

$V$ wants to vote for $0$ in both elections. Since the vote is for the same candidate, $V$ (unwisely) decides to reuse the same term, instead of re-encrypting in a fresh key. So we have a run where $V$ sends both $\exists xr.\bigl\{ \equals{\{x\}_{r}}{\{v\}_{k}} \conj x \listmemb [0, 1] \bigr\}$ and $\exists ys.\bigl\{ \equals{\{y\}_{s}}{\{v\}_{k}} \conj y \listmemb [0, 2] \bigr\}$. Now, since the same term $\{v\}_{k}$ is involved in both assertions, an observer ought to be able to deduce that the vote is actually for $0$. This would allow them access to both the identity of a voter as well as their vote, falsifying anonymity. The assertion system formally captures such inference via a proof system.
\end{example}

\subsection{Abstractability and Proof System}
Before we present the proof system, we need to fix under what conditions one can derive a new assertion from existing ones. In a security context, it becomes important to distinguish when a term is accessible inside an assertion versus when it is not. To substitute a term $u$ (with, say, $v$) inside a term $t$, an agent $A$ essentially needs to break the term down to that position, replace $u$ with $v$, and construct the whole term back. This depends on other terms $A$ has access to. We formalize this notion as ``abstractability'', which requires us to first define the set of term positions of an assertion. 

We will view terms as trees, with $\positions{t} \subseteq \nat^{*}$ denoting the set of positions of the term $t$, and $\epsilon$ the empty word in $\nat^{*}$. We will also view assertions as trees, with any operator forming the root of its subtree, and its operands standing for its children. We will only be interested in the position where terms occur in assertions, not those of the various operators. We define these as follows.

\begin{definition}[Term positions of an assertion]
    We define the \defemph{term positions of an assertion} $\alpha$, denoted $\positions{\alpha}$, as follows:
	\begin{itemize}
        \item $\positions{\equals{t}{t'}} = \{0\cdot p \mid p \in \positions{t}\} \cup \{1\cdot p \mid p \in \positions{t'}\}$
        \item $\positions{P(u_{0}, \ldots, u_{m})} = \{0, \ldots, m\}$
        \item $\positions{t \listmemb [t_{1}, \ldots, t_{n}]} = \{0, 1, \ldots, n\}$
        \item $\positions{\alpha \conj \beta} = \{0\cdot p \mid p \in \positions{\alpha}\} \cup \{1\cdot p \mid p \in \positions{\beta}\}$
	    \item $\positions{\exists{x}.\alpha} = \{0\cdot p \mid p \in \positions{\alpha}\}$				
		\item $\positions{\pk(k) \says \alpha} = \{0, 00\} \cup \{1\cdot p \mid p \in \positions{\alpha}\}$
	\end{itemize}    
\end{definition}

For $t, r \in \Tterms$, and $p \in \positions{t}$, $\subtermat{t}{p}$ is the subterm of $t$ rooted at $p$. The set of positions of $r$ in $t$ is $\posof{r}{t} \coloneqq \{p \in \positions{t} \mid \subtermat{t}{p} = r\}$. For $P \subseteq \positions{t}$, $\replsubtermat{t}{P}{r}$ is obtained by replacing the subterm of $t$ occurring at each $p \in P$ with $r$. We will use analogous notation for assertions. 

\begin{definition}[Abstractable positions of a term]
    Let $S \cup \{t\} \subseteq \Tterms$. The set of \defemph{abstractable positions of $t$ w.r.t.\ $S$}, denoted $\abstractable(S, t)$, is defined as follows. For $p \in \positions{t}$, let $\mathbb{Q}_{p} = \{\varepsilon\} \cup \{qi \in \positions{t} \mid q$ is a proper prefix of $p\}$. Then $\abstractable(S,t) \coloneqq \{p \in \positions{t} \mid S \DYderives \subtermat{t}{q}$ for all $q \in \mathbb{Q}_{p}\}$.    
\end{definition}

For example, let $t = (\{\{m\}_{k}\}_{k'}, (n_{1}, n_{2}))$. Then,
$\positions{t} = \{\epsilon, 0, 1, 00, 01, 10, 11, 000, 001\}$. Consider the set $S = \{\{m\}_{k}\}_{k'}, (n_{1}, n_{2})$. Then, $\abstractable(S, t) = \{\varepsilon, 0, 1, 10, 11\}$. The abstractable positions are shown in bold in Figure~\ref{fig:absterms}. 

\begin{figure}[ht]
\centering
\begin{tikzpicture}[baseline,level distance = 8mm]
  \tikzstyle{every node}=[draw=white,text=black]
   \tikzstyle{level 1}=[sibling distance=30mm]
  \tikzstyle{level 2}=[sibling distance=20mm]
  \tikzstyle{level 3}=[sibling distance=10mm]
  \node(e) {\pair}
	child 
	{ 
		node(0) {\senc} 
		child
		{
	      	node(00) {$\senc$}
			child
			{
				node(000) {$m$}
				node[below=0.1mm of 000, draw=black]{000}	
			}
			child
			{
				node(001) {$k$}
				node[below=0.1mm of 001, draw=black]{001}
			}
			node[left = 0.05mm of 00, draw=black]{00}
		}
		child
		{
		node(01) {$k'$}
		node[below = 0.1mm of 01, draw=black]{01}
		}
		node[left = 0.05mm of 0, draw=black, line width=0.5mm]{0}
	} 	
	child
	{
	      node(1) {$\pair$}
	      child
	      {
	      	node(10) {$n_{1}$}
			node[below=0.1mm of 10, draw=black, line width=0.5mm]{10}
	      }
	      child
	      {
	      	node(11) {$n_{2}$}
			node[below=0.1mm of 11, draw=black, line width=0.5mm]{11}
	      }
	      node[right = 0.05mm of 1, draw=black, line width=0.5mm]{1}
	}	   
	node[above = 0.1mm of e, draw=black, line width=0.5mm]{$\varepsilon$}
;  
\end{tikzpicture}
\caption{Abstractable positions w.r.t.~$S = \{\{m\}_{k}\}_{k'}, (n_{1}, n_{2})\}$}
\label{fig:absterms}
\end{figure}
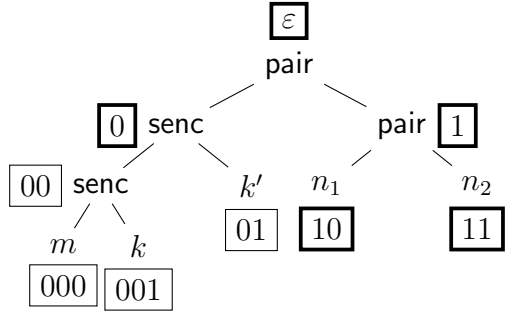

Now, an inductive definition seems like it might suffice to lift the notion of abstractable positions for assertions. However, a problem arises when we consider an assertion of the form $\exists x.\alpha$. Let $\alpha = \exists b.\{\equals{\{m\}_{b}}{\{m\}_{k}}\}$. Suppose we want to get $\exists ab.\{\equals{\{a\}_{b}}{\{m\}_{k}}\}$ from $\alpha$ in the presence of the set $S = \{m, k\}$. That position of $m$ in $\alpha$ must be abstractable w.r.t~$S$, i.e. we require that $S \DYderives \{m\}_{b}$, but $S$ does not even contain the quantified variable $b$. We must therefore consider derivability from $S \cup \{b\}$ in this case, not $S$.

\begin{definition}[Abstractable positions of an assertion]\label{def:assabs}
	The set of \emph{abstractable positions of $\alpha$ w.r.t.\ $S$}, denoted by $\abstractable(S,\alpha)$, is:
	\begin{itemize}
		\item $\abstractable(S, \equals{t_{0}}{t_{1}}) = \{i\cdot p \mid i \in \{0, 1\}, \ p \in \abstractable(S,t_{i})\}$
		\item $\abstractable(S, P(u_{1}, \ldots, u_{m})) = \{i \mid 1 \leq i \leq m, S \DYderives u_{i}\}$
		\item $\abstractable(S, t \listmemb [t_{1}, \ldots, t_{n}]) = \{0\}$
		\item $\abstractable(S, \alpha_{0}\conj\alpha_{1}) = \{i\cdot p \mid i \in \{0, 1\}, \ p \in \abstractable(S,\alpha_{i})\} $
		\item $\abstractable(S, \exists{x}.\alpha) = \{0\cdot p \mid p \in \abstractable(S \cup \{x\},\alpha) \}$
		\item $\abstractable(S, \pk(k) \says \alpha) =  \{0\} \cup \{1\cdot p \mid p \in \abstractable(S, \alpha)\}$
	\end{itemize}
\end{definition}

We now state a fundamental property of abstractability, which will be used in some of the more technical proofs later.

\begin{lemma}\label{lem:absderiv}
Let $S \cup \{t,r\} \subseteq \Tterms$ s.t.\ $S \DYderives r$. If $x \notin \varsof(S)$ and $P = \posof{x}{t} \subseteq \abstractable(S \cup \{x\}, t)$, then $\abstractable(S, \replsubtermat{t}{P}{r}) \cap \positions{t} = \abstractable(S \cup \{x\}, t)$. 
\end{lemma}
\begin{proof}
	For any term $a$ and any set $Q \subseteq \positions{a}$, we let $\subtermat{a}{Q}$ denote $\{\subtermat{a}{q} \mid q \in Q\}$. We now observe some general properties of abstractability. 
	
	For any $T, a$ and $q \in \abstractable(T,a)$ s.t.\ $\subtermat{a}{q}$ is non-atomic, either $\{q0, q1\} \subseteq \abstractable(T,a)$ and $\subtermat{a}{\{q0,q1\}} \DYderives \subtermat{a}{q}$ via a constructor rule, or $q$ is a maximal position in $\abstractable(T,a)$ (it is not the prefix of any other position in the set). We have the following two properties.
\begin{enumerate}
\item Let $M = \{q \in \positions{a} \mid q$ is a maximal position in $\abstractable(T,a)\}$. Then for every $p \in \abstractable(T,a)$, $\subtermat{a}{M} \DYderives \subtermat{a}{p}$ via a proof consisting only of constructor rules. 
\item\label{item:prefsib} Suppose $Q \subseteq \positions{a}$ is prefix-closed (if $q \in Q$ and $p$ is a prefix of $q$, then $p \in Q$) and sibling-closed (if $qi \in Q$ and $qj \in \positions{a}$, then $qj \in Q$). If $T \DYderives \subtermat{a}{q}$ for every maximal $q \in Q$, then $Q \subseteq \abstractable(T,a)$. 
\end{enumerate}
		 
    We now prove the statement of the lemma. Let $u = \replsubtermat{t}{P}{r}$, and let $A$ and $B$ denote $\abstractable(S\cup\{x\},t)$ and $\abstractable(S,u) \cap \positions{t}$ respectively. Note that $A$ and $B$ are both prefix-closed and sibling-closed. Let $M$ (resp.\ $N$) be the set of maximal positions in $A$ (resp.\ $B$). 
    
    Since $P \subseteq A$ is the set of $x$-positions in $t$, $P \subseteq M$ and no $q \in M$ is a prefix of a position in $P$. Thus, for every $q \in M$, either $\subtermat{t}{q} = x$, or $x \notin \varsof(\subtermat{t}{q})$. If $\subtermat{t}{q} = x$, $\subtermat{u}{q} = r$, and $S \DYderives \subtermat{u}{q}$ (since $S \DYderives r$). If $x \notin \varsof(\subtermat{t}{q})$, then $\subtermat{u}{q} = \subtermat{t}{q}$ and $S \DYderives \subtermat{u}{q}$. This is because $q \in \abstractable(S\cup\{x\},t)$, so $S \cup \{x\} \DYderives \subtermat{t}{q}$, but $x$ does not occur in the conclusion. Thus we have $S \DYderives \subtermat{u}{q}$ for every $q \in M$. Since $A$ is prefix-closed and sibling-closed, by~\ref{item:prefsib}, we get $A \subseteq \abstractable(S,u)$. Since $A \subseteq \positions{t}$ as well, we get $A \subseteq B$. 
    
	By similar reasoning as above, we can see that $S \cup \{x\} \DYderives \subtermat{t}{q}$ for each $q \in N$. (For some of these positions $q$, $x$ does not occur at all in the subterm at that position, and $\subtermat{t}{q} = \subtermat{u}{q}$ is derivable from $S$. For other positions $q$, $\subtermat{t}{q} = x$ and is derivable from $S \cup \{x\}$.) Therefore $B \subseteq A$. 
\end{proof}

The assertion proof system is shown in Table~\ref{tab:asstheory}. We say $S; A \assderives \alpha$ if $\alpha$ can be derived from $S; A$ using these rules. We say $S ; A \assderives \Gamma$ if $S; A \assderives \gamma$ for every $\gamma \in \Gamma$. 
\begin{table*}[!t]
	\centering {\scriptsize
        \tabulinesep=2mm
        \setlength{\tabcolsep}{0.6em}
        \begin{tabu}{|c|c|c|}
            \hline
			\begin{math}
                {
                    \begin{prooftree}
                        \justifies S; A \cup \{\alpha\} \vdash \alpha \using \rnax
				\end{prooftree}
			 }
			\end{math}
			&
			\begin{math}
                {
					\begin{prooftree}
						S \DYderives t
						\justifies S; A \vdash \equals{t}{t} \using \rneq
					\end{prooftree}
                }
			\end{math}
			 &
			\begin{math}
			{
				\begin{prooftree}
				S; A \vdash \equals{t_{0}}{u_{0}} \quad S; A \vdash \equals{t_{1}}{u_{1}}
				\justifies S; A \vdash \equals{\func(t_{0}, t_{1})}{\func(u_{0}, u_{1})} \using \rncons
				\end{prooftree}
			}
			\end{math}
			\\
			\hline
			\begin{math}
                {
                    \begin{prooftree}
                        S; A \vdash \equals{t}{u}
                        \justifies S; A \vdash \equals{u}{t} \using \rnsymm
                    \end{prooftree}
                }
            \end{math}
            &
            \begin{math}
                {
                    \begin{prooftree}
                         S; A \vdash \equals{t_{1}}{t_{2}} \cdots S; A \vdash \equals{t_{k}}{t_{k+1}}
                        \justifies S; A \vdash \equals{t_{1}}{t_{k+1}} \using \rntrans
                    \end{prooftree}
                }
            \end{math}
            &
            \begin{math}
                {
                    \begin{prooftree}
                        S; A \vdash \equals{\func(t_{1}, \ldotp \ldotp, t_{r})}{\func(u_{1}, \ldotp \ldotp, u_{r})}
                        \justifies S; A \vdash \equals{t_{i}}{u_{i}} \using \rnproj_{i}^{\P}
					\end{prooftree}
                }
            \end{math}
            \\ 
            \hline 	
            \begin{math}
                {
					\begin{prooftree}
						S; A \vdash \alpha_{0} \quad S; A \vdash \alpha_{1}
						\justifies S; A \vdash \alpha_{0} \conj \alpha_{1} \using \rnconji
					\end{prooftree}
                }
			\end{math}
			&
			\begin{math}
                {
					\begin{prooftree}
						S; A \vdash \alpha_{0} \conj \alpha_{1}
						\justifies S; A \vdash \alpha_{i} \using \rnconje_{i}
					\end{prooftree}
                }
			\end{math}
			&
			\begin{math}
                {
                    \begin{prooftree}
                        S; A \vdash t \listmemb l \quad S; A \vdash \equals{t}{u}
                        \justifies S; A \vdash u \listmemb l \using \rnsubst
					\end{prooftree}
                }
			\end{math}
			\\
			\hline
			\begin{math}
                {
                    \begin{prooftree}
                        S; A \vdash \replsubtermat{\alpha}{P}{t} \quad S \DYderives t
                        \justifies S; A \vdash \exists{x}.\alpha \using \rnexintro^{\ddag}
                    \end{prooftree}
                }
            \end{math}
			&
            \begin{math}
                {
                    \begin{prooftree}
                        S; A \vdash \exists{}{x}.\alpha \quad S \cup \{y\}; A \cup \{\replsubtermat{\alpha}{P}{y}\} \vdash \gamma
                        \justifies S; A \vdash \gamma \using \rnexelim^{\S}
                    \end{prooftree}
                }
            \end{math}
			&
            \begin{math}
			{
				\begin{prooftree}
					S; A \vdash \alpha \quad S \DYderives k
					\justifies S; A \vdash \pk(k) \says \alpha \using \rnsays	
				\end{prooftree}
			}
			\end{math}
            \\
            \hline 
			\begin{math}
                {
                    \begin{prooftree}
                        S; A \vdash t \listmemb [n]
                        \justifies S; A \vdash \equals{t}{n} \using \rnprom
					\end{prooftree}
                }
			\end{math}
			&
			\begin{math}
                {
                    \begin{prooftree}
                        S; A \vdash t \listmemb l_{1} \ldots S; A \vdash t \listmemb l_{m}
                        \justifies S; A \vdash t \listmemb (l_{1} \cap \ldots \cap l_{m}) \using \rnlint
					\end{prooftree}
                }
			\end{math}	
			&
			\begin{math}
                {
                    \begin{prooftree}
                        S; A \vdash \equals{t}{n_{i}} \quad S \DYderives n_{i} (\forall i \leq n)
                        \justifies S; A \vdash t \listmemb [n_{1}, \ldots, n_{k}] \using \rnlwk
					\end{prooftree}
                }
			\end{math}
			\\
			\hline
		\end{tabu}
	}
	\caption{
        Derivation system $\assderives$ for assertions. 
        $\P$ states that $\{0i, 1i \mid i \le r\} \subseteq \abstractable(S, \equals{\func(t_{1}, \ldots, t_{r})}{\func(u_{1}, \ldots, u_{r}))}$. $\dag$ demands that $P \subseteq \posof{x}{\alpha} \cap \abstractable(S \cup \{x\}, \alpha)$, and no position in $P$ occurs in the scope of a$\says$\!\!. $\ddag$ stands for $P = \posof{x}{\alpha} \subseteq \abstractable(S \cup \{x\}, \alpha)$. $\S$ states that $y \notin \freevars(S) \cup \freevars(A) \cup \freevars(\gamma)$ and $P = \posof{x}{\alpha}$. 
    }
	\label{tab:asstheory}
\end{table*}

We say that $S;A \eqderives \alpha$ if $\alpha$ can be derived from $S;A$ by a proof which does not use any of the rules from $\{\rnconji, \rnconje, \rnexintro, \rnexelim, \rnsays\}$. Recall that an atomic assertion is one that is not of the form $\alpha\conj\beta$ or $\exists{x}.\alpha$. The $\eqderives$ system is used typically when $A \cup \{\alpha\}$ consists only of atomic assertions, and we want to ensure that there is no use of the rules for $\conj$ and $\exists$ in these proofs. To ensure this, we also need to avoid the $\rnsays$ rule. Otherwise, we might allow a derivation of $\pk(k)\says(\alpha\conj\beta)$ using $\alpha\conj\beta$, which itself can be derived only using $\rnconji$ (since the LHS contains only atomic assertions).

The proofs in Section~\ref{sec:insecurity} crucially appeal to some properties of $\eqderives$ proofs, which we detail below.

\begin{restatable}{definition}{eqnormal}\label{def:eqnormal}
Suppose $E \cup \{\alpha\}$ consists only of atomic formulas and $\pi$ is a proof of $T; E \eqderives \alpha$. We use ``$\rnrule_{1}$ \emph{precedes} $\rnrule_{2}$ in $\pi$'' to mean that the conclusion of some application of $\rnrule_{1}$ is a premise of an application of $\rnrule_{2}$ in $\pi$.

We say that $\pi$ is \emph{normal} if the following hold.
	\begin{enumerate}
		\item All $\DYderives$ subproofs are normal. 
		\item $\rnsymm$ can only be preceded by $\rnax$ or $\rnprom$.
		\item $\rneq$ can only be preceded by a destructor rule.
		\item No premise of a $\rntrans$ is of the form $\equals{a}{a}$, or the conclusion of a $\rntrans$.
		\item Adjacent premises of a $\rntrans$ are not conclusions of $\rncons$.
		\item $\rnlint$ cannot be preceded by $\rnlint$ or $\rnlwk$.
		\item No subproof ending in $\rnproj$ contains $\rncons$. 
\end{enumerate}\end{restatable}

We now state the normalization theorem and subterm property for $\eqderives$ proofs. First, we define the following notions.
\begin{itemize}[leftmargin=*]
\item $\termsof(\pi) \coloneqq \{t \mid$ a subproof of $\pi$ derives $\alpha$ and $t$ is a maximal subterm of $\alpha\}$. 
\item $\listsof(E) \coloneqq \{\ell \mid \exists{t}:t\listmemb{\ell}$ is in $E\}$.
\item $\listsof(\pi) \coloneqq \{\ell \mid$ a subproof of $\pi$ derives $t\listmemb{\ell}\}$. 
\end{itemize}

\begin{theorem}[Normalization \& Subterm Property for $\eqderives$]\label{thm:eqnormsub}
\phantom{a}
\begin{enumerate}[leftmargin=*]
\item If $(T;E) \eqderives \alpha$ then there is a normal proof of $(T;E) \vdash \alpha$ in the $\eqderives$ system.
\item For any normal proof $\pi$ of $T; E \eqderives \alpha$, letting $Y = \subterms(T) \cup \subterms(E\cup\{\alpha\})$, we have: 
\begin{itemize}
    \item $\termsof(\pi) \subseteq Y$.
    \item $\listsof(\pi) \subseteq \listsof(E \cup \{\alpha\}) \cup \{[n] \mid n \in Y\}$.
\end{itemize}
\end{enumerate}
\end{theorem}

Armed with these notions, we present a saturation-based procedure in Algorithm~\ref{alg:eqderive} for deciding whether $T; E \eqderives \alpha$, where $E \cup \{\alpha\}$ consists only of atomic assertions. The procedure computes the set \[\newconseq{T}{E}{\alpha} \coloneqq \bigl\{\beta \mid \beta \text{ is atomic}, \beta \in Z, (T; E) \eqderives \beta\bigr\}\] 
where $Z$ is as defined in Algorithm~\ref{alg:eqderive}, and checks if $\alpha \in \newconseq{T}{E}{\alpha}$.

Letting $M = |\subterms(T) \cup \subterms(E \cup \{\alpha\})|$ and $N = |\listsof(E)|$, it can be seen that the algorithm runs in time polynomial in $M+N$. There are at most $(M+N)^{2}$ atomic formulas that can be added in $C$, and hence the \textbf{while} loop runs for at most $(M+N)^{2}$ iterations. In each iteration, the amount of work to be done is polynomial in $M+N$. (Recall that $\DYderives$ can be decided in PTIME.) Thus the algorithm works in time polynomial in $M+N$, and hence polynomial in the size of $(T; E \cup \{\alpha\})$.

\begin{algorithm}[t]
	\caption{Algorithm to compute $\newconseq{T}{E}{\alpha}$, given $(T; E), \alpha$}
	\label{alg:eqderive}
	\begin{algorithmic}[1]
		\State $Y \gets \subterms(S) \cup \subterms(E \cup \{\alpha\})$;
		\State $Z \gets \big\{\beta \mid \beta \text{ is atomic}, \subterms(\beta) \in Y,$ 
		\State \quad \quad \quad $\listsof(\beta) \subseteq \listsof(E) \cup \{[n] \mid n \in Y\}\big\};$
		\State $B \gets \emptyset$;
		\State $C \gets E$;
		\While{$(B \neq C)$}
		\State $B \gets C$;
		\State $C \gets B \cup \big\{\beta \in Z \mid \beta \text{ can be obtained from $B$ using}$ 
		\State \qquad \qquad \qquad $\text{one application of any rule in $\assderives$}\big\}$;
		\EndWhile
		\State \Return $B$.
	\end{algorithmic}
\end{algorithm}

\subsection{Protocols and runs}
Following~\cite{Bla01, RT03}, a \emph{protocol} is given by a finite set of roles, each role consisting of a finite sequence of alternating receives and sends (each send triggered by a receive).\!\footnote{This model considers actions of the form $\assertact~\alpha$ to model rudimentary branching in protocols, which we used for specifying the FOO protocol. But we omit these in the formal model, for ease of presentation. We discuss handling such branching in Section~\ref{sec:ifthenelse}.} These are the actions of \emph{honest agents}. Every sent message is added to the Dolev-Yao intruder's knowledge base. Each received message is assumed to have come from the intruder, so it must be derivable by the intruder. We assume that only assertions are communicated -- a term $t$ can be modelled via the assertion $\equals{t}{t}$, whose only public term is $t$. 

A \textit{protocol} $\prot$ is a finite set of \textit{roles}, each of the form $({\beta_{1}},{\alpha_{1}})\ldots({\beta_{m}},{\alpha_{m}})$, where the $\alpha_{i}$s and $\beta_{i}$s are assertions. An $x \in \freevars(\prot)$ is said to be an \emph{agent variable} if it occurs first in an $\alpha_{i}$; otherwise it is an \emph{intruder variable}. Each role is a sequence of actions by an agent, receiving the $\beta_{i}$s and sending the $\alpha_{i}$s in response. The $\alpha_{i}$s and $\beta_{i}$s can have bound variables from $\qvar$ as well as free variables from $\ivar$. Instantiating the free variables with appropriately-typed ground terms yields a \emph{session}. A \emph{run} is obtained by interleaving a finite number of sessions that satisfy the required derivability conditions. It is convenient to instantiate the free variables of a role in two stages. Agent variables are instantiated with names before starting a session, but intruder variables can be mapped to terms only at runtime. 

A \emph{session} of a protocol $\prot$ is a sequence of the form $u:\recsend{\beta_{1}}{\alpha_{1}}\ \cdots\ u:\recsend{\beta_{\ell}}{\alpha_{\ell}}$ where $u \in \ag$ and $(\beta_{1}, \alpha_{1}) \cdots (\beta_{\ell}, \alpha_{\ell})$ is a prefix of a role of $\prot$ with all the agent variables instantiated by values from $\names$. A set of sessions $S$ of $\prot$ is \emph{coherent} if $\freevars(\xi) \cap \freevars(\xi') = \emptyset$ for distinct $\xi, \xi' \in S$. One can always achieve coherence by renaming intruder variables as necessary. 

A run is an interleaving of sessions where each message sent by an agent should be constructible from their knowledge. 
A \emph{knowledge state} is a pair $(X; \Phi)$ where $X$ is a finite set of terms and $\Phi$ is a finite set of assertions. A \emph{knowledge function} $\knowfunc$ is such that $\dom(\knowfunc) = \ag$ and for each $a \in \ag$, $\knowfunc(a)$ is a knowledge state. 

Given a knowledge state $(X; \Phi)$ and an assertion $\alpha$, we define $\knowupd((X;\Phi), \alpha) \coloneqq (X \cup \publics(\alpha), \Phi \cup \{\alpha\})$.

\begin{definition}\label{def:protruns}
    A \emph{run} of a protocol $\prot$ is a pair $(\xi, \intsub)$ where:
    \begin{itemize}[leftmargin=*]
		\item $\xi \coloneqq u_{1}:\recsend{\beta_{1}}{\alpha_{1}}, \ldots, u_{n}:\recsend{\beta_{n}}{\alpha_{n}}$ is an interleaving of a finite, coherent set of sessions of $\prot$.
        \item $\intsub$ is a ground substitution with $\dom(\intsub) = \freevars(\xi)$.
        \item There is a sequence $\knowfunc_{0} \ldots \knowfunc_{n}$ of knowledge functions s.t.:
        \begin{itemize}
            \item $\knowfunc_{0}(a) = (X_{a}; \emptyset)$, where $X_{a}$ is a finite set of initial terms known to $a$ ($a$'s secret key, public keys, public names etc). 
            \item For all $i < n$, 
            \[
                \knowfunc_{i+1}(a) = 
                \begin{cases}
                    \knowfunc_{i}(a) & \mbox{if $a \neq u_{i}, a \neq I$} \\ 
                    \knowupd(\knowfunc_{i}(a), \beta_{i}) & \mbox{if $a = u_{i}$} \\
                    \knowupd(\knowfunc_{i}(a), \alpha_{i}) & \mbox{if $a = I$}
                \end{cases}
            \]
            \item For $i \leq n$, $\knowfunc_{i}(u_{i}) \assderives \alpha_{i}$ and $\intsub(\knowfunc_{i-1}(I)) \assderives \intsub(\beta_{i})$.
        \end{itemize}
    \end{itemize}
\end{definition}

Note that honest agent derivations of the form $\knowfunc_{i}(u_{i}) \assderives \alpha_{i}$ do not depend on accidental unification with intruder variables under $\intsub$; rather, they hold even in the ``abstract''. 

We can write an $A$-session and a $B$-session for the Example~\ref{ex:exprot} protocol as $A : \beta_{1} \Rightarrow \alpha_{1}, A : \beta_{3} \Rightarrow \alpha_{3}$ and $B : \beta_{2} \Rightarrow \alpha_{2}$. (To save space, we denote by $p_{A}$ and $p_{B}$ the keys $\pk(k_{A})$ and $\pk(k_{B})$.) We assume that $A$ starts a session by receiving a dummy name $s$, and ends the session by sending $s$ out, and code up each communicated term $t$ from Example~\ref{ex:exprot} as the assertion $\equals{t}{t}$. Note that $A, B, p_{A}, m$, and $p_{B}$ are names used to instantiate agent variables in these sessions. The set of these two sessions is coherent. 
    
\begin{math}
\beta_{1} = \equals{s}{s} \qquad \qquad \hspace{3mm} \alpha_{1} = \{\equals{(p_{A}, \{m\}_{p_{B}})}{(p_{A}, \{m\}_{p_{B}})} \} \\
\beta_{2} = \{\equals{(x, \{y\}_{p_{B}})}{(x, \{y\}_{p_{B}})} \} \qquad \hspace{2mm}  \alpha_{2} = \{\equals{\{y\}_{x}}{\{y\}_{x}}\} \\
 \beta_{3} = \{\equals{\{m\}_{p_{A}}}{\{m\}_{p_{A}}} \} \qquad \qquad \hspace{5mm} \alpha_{3} = \equals{s}{s}
\end{math}

Consider the substitution $\sigma = [x \mapsto p_{A}, y \mapsto m]$ applied to $\xi = A:\beta_{1} \Rightarrow \alpha_{1}, B:\beta_{2} \Rightarrow \alpha_{2}, A:\beta_{3} \Rightarrow \alpha_{3}$. This would be a run $(\xi, \sigma)$ where the intruder just observes traffic on the network, but does not interfere otherwise. 

Let $X_{B} = \{ A, B, p_{A},  p_{B}, k_{B} \}$. $\knowfunc_{0}(B) = (X_{B}; \emptyset)$. Note that $\knowfunc_{1}(B) = \knowfunc_{0}(B)$. There is an update to $B$'s knowledge state only upon receipt of $\beta_{2}$. So, $\knowfunc_{2}(B) = \knowupd(\knowfunc_{1}(B), \beta_{2})$ is given by $(X'; \Phi)$ where $X' = X \cup \{(p_{A}, \{m\}_{p_{B}})\}$ and $\Phi = \{ \{\equals{(p_{A}, \{m\}_{p_{B}})}{(p_{A}, \{m\}_{p_{B}})} \}$.

We can also consider a run with the same $\xi$ under a substitution $\sigma = [x \mapsto \pk(k_{I}), y \mapsto m]$, which represents the man-in-the-middle attack shown earlier.

A \emph{secrecy property} is given by an assertion $\gamma$ that the intruder should not know. A \emph{$K$-bounded attack} which violates the secrecy of $\gamma$ is a run of the protocol with at most $K$ sessions where $\intsub(\knowfunc_{n}(I)) \assderives \intsub(\gamma)$.

\begin{definition}[$K$-bounded insecurity problem]
    Given a protocol $\prot$ and a designated assertion $\gamma$, check whether there exists a $K$-bounded attack on $\prot$ violating the secrecy of $\gamma$. 
\end{definition}
Henceforth, we will use ``insecurity problem'' to mean the $K$-bounded insecurity problem for some $K$.

\section{Proof strategy for the insecurity problem}\label{sec:proofstrat}
In the subsequent sections, we will show that the $K$-bounded insecurity problem for assertions is in NP. But first, we provide an overview of the proof strategy we will employ.

Given a protocol $\prot$, a secrecy property specified by an assertion $\gamma$ and a bound $K$ (in unary), one way to check if there is a $K$-bounded attack works as follows: Guess a coherent set of sessions of size $K$, an interleaving $\xi = u_{1}:\recsend{\beta_{1}}{\alpha_{1}}, \ldots, u_{n}:\recsend{\beta_{n}}{\alpha_{n}}$, and a substitution $\intsub$ with $\dom(\intsub) = \freevars(\xi)$, and check that $(\xi, \intsub)$ satisfies the conditions in Definition~\ref{def:protruns}. For this, we need an effective check for derivabilities of the form $\intsub(\knowfunc_{i-1}(I)) \assderives \intsub(\beta_{i})$.

As with terms, this needs us to bound the size of terms assigned to variables by $\intsub$. However, we also have quantified variables in our proofs, for which witnesses need to be assigned. To check whether a formula of the form $\exists{x}.~\alpha$ is derivable, one would in general have to check if $\alpha(t)$ is derivable for some $t$, which might be unboundedly large. To get an effective algorithm, we have to show that if there is a witness at all, there is a witness of small size.  

One way to represent these witnesses is via a substitution $\mu$ which maps each quantified variable $x$ to the appropriate witness. To obtain small witnesses, we adapt the techniques of~\cite{RT03}. For this, it is helpful to first simplify the LHS to contain only atomic formulas. Any normal proof of $\alpha$ from such an LHS will not involve $\rnconje$ or $\rnexelim$. We further show, via Theorem~\ref{thm:eqs-to-gamma}, that these proofs can be decomposed into multiple proofs, one for each atomic subformula of $\alpha$ (with witnesses instantiated by $\mu$), and then applying $\rnconji$ and $\rnexintro$. 

Applying Theorem~\ref{thm:eqs-to-gamma} to each derivability check $\intsub(\knowfunc_{i-1}(I)) \assderives \intsub(\beta_{i})$ for $1 \le i \le n$, we get a set of witness substitutions $\{\iwsub_{1}, \dots, \iwsub_{n}\}$. We would like to ensure that all of these, along with $\intsub$, can be chosen to be ``small''. 

In order to obtain these small substitutions, we follow the techniques of~\cite{RT03}. This involves identifying and mapping to atomic terms variables that do not map to any term that ``corresponds'' to one in the protocol specification. However, unlike~\cite{RT03}, we need to do this simultaneously for multiple substitutions -- $\intsub$ (which instantiates intruder variables) and $\iwsub_{i}$ (which instantiates quantified variables). The various $\iwsub_{i}$s might be influenced by $\intsub$, so preserving derivabilities when moving to small substitutions becomes a challenge. In order to do this, we employ a notion of ``typed proofs'', both for the $\DYderives$ and $\eqderives$ systems. We show that any proof can be converted to a typed equivalent, and typed proofs make it easier for us to replace the substitutions therein with small ones while preserving derivations.

We will now present the solution in detail.

\section{Solving the insecurity problem for \texorpdfstring{$\assderives$}{assertion derivability}}\label{sec:insecurity}

We fix a protocol $\prot$ and a run $(\xi, \intsub)$ of $\prot$. By renaming variables if necessary, we can ensure that $\freevars(\xi) \cap \qvar = \emptyset$. Thus, in all proof sequents that we consider, no variable has both free and bound occurrences, and no variable is quantified by distinct quantifiers. Furthermore, whenever we use $(S; A)$, we mean that $S$ is a set of terms, $A$ is a set of assertions, and $S$ derives the public terms of all assertions in $A$. 

We also use $\varsof(S;A)$ to mean $\varsof(S) \cup \varsof(A)$ and $\freevars(S;A)$ to mean $\varsof(S) \cup \freevars(A)$. 

As a first step, we move to an LHS consisting solely of atomic formulas. For this, we will employ the following two ``left'' properties enjoyed by the $\assderives$ system.

\begin{lemma}\label{lem:leftprop}
\phantom{a}
\begin{enumerate}
\item $(S; A \cup \{\alpha\conj\beta\}) \assderives \gamma$ iff $(S; A\cup\{\alpha,\beta\}) \assderives \gamma$.
\item Let $S, A, \exists x.\alpha$ and $\gamma$ be such that $x \notin \varsof(S) \cup \varsof(A \cup \{\gamma\})$ and $\posof{x}{\alpha}\subseteq \abstractable(S \cup \{x\}, \alpha)$. Then $(S; A \cup \{\exists x.\alpha\}) \assderives \gamma$ iff $(S \cup \{x\}; A \cup \{\alpha\}) \assderives \gamma$.
\end{enumerate}
\end{lemma}
\begin{proof}
\begin{enumerate}
\item
    To save space, we use $A, \phi$ to mean $A \cup \{\phi\}$ in the proof to follow. 

    For the left to right direction, let $\pi$ be a proof of $S; A,\alpha\conj\beta \vdash \gamma$. The following is a proof of $S; A,\alpha,\beta \vdash \gamma$.
    \[
        \begin{prooftree}
            \[
                \[
                    \justifies S; A,\alpha,\beta \vdash \alpha \using \rnax
                \]
                \[
                    \justifies S; A,\alpha,\beta \vdash \beta \using \rnax
                \]
                \justifies S; A,\alpha,\beta \vdash \alpha\conj\beta \using \rnconji
            \]
            \[
                \pi \leadsto S; A,\alpha\conj\beta \vdash \gamma
            \]
            \justifies S; A,\alpha,\beta \vdash \gamma
        \end{prooftree}
    \]
    For the other direction, let $\pi$ be a proof of $S; A,\alpha,\beta \vdash \gamma$. We obtain a proof of $S; A,\alpha\conj\beta \vdash \gamma$ below. We omit the $S; A$ part of the LHS to conserve space.
    \[
        \begin{prooftree}
            \[
                \[
                    \justifies \alpha\conj\beta \vdash \alpha\conj\beta \using \rnax
                \]
                \justifies \alpha\conj\beta \vdash \beta \using \rnconje{}
            \]
            \[
                \[
                    \[
                        \justifies \alpha\conj\beta \vdash \alpha\conj\beta \using \rnax
                    \]
                    \justifies \alpha\conj\beta \vdash \alpha \using \rnconje{}
                \]
                \[
                    \pi \leadsto \alpha,\beta \vdash \gamma
                \]
                \justifies \alpha\conj\beta,\beta \vdash \gamma \using
            \]
            \justifies \alpha\conj\beta \vdash \gamma
        \end{prooftree}
    \]
    We have freely used the \emph{cut rule}, which is admissible in our system.  
    \[ 
        \begin{prooftree}
            S; A \vdash \phi \qquad S; B,\phi \vdash \psi 
            \justifies S; A \cup B \vdash \psi 
        \end{prooftree}    
    \]
    If $\pi_{0}$ and $\pi_{1}$ are derivations of the left and right premises as above, then we can replace each axiom rule occurring in $\pi_{1}$ and deriving $\phi$, with the proof $\pi_{0}$, thus yielding a proof of $S; A \cup B \vdash \psi$. 

\item
    For the left to right direction, let $\pi$ be a proof of $S; A, \exists{x}.\alpha \vdash \gamma$. Note that we have a proof $\pi_{1}$ of $\exists{x}.\alpha$ from $(S, x; A, \alpha)$, where the $\rnexintro$ rule is justified because the abstractability side condition $\posof{x}{\alpha} \subseteq \abstractable(S \cup \{x\}, \alpha)$ is assumed. We can then use the $\rncut$ rule (which is admissible in $\assderives$) on this proof along with the proof $\pi$ to get $(S, x; A, \alpha) \assderives \gamma$. 
    \[
        \begin{prooftree}
        \[
            \[
                \justifies S,x; A,\alpha \vdash \alpha \using \rnax
            \]
            \justifies S,x; A,\alpha \vdash \exists{x}.\alpha \using \rnexintro
        \]
        \[
				\pi \leadsto
        	\justifies S; A, \exists{x}.\alpha \vdash \gamma
        \]
        \justifies S, x; A, \alpha \vdash \gamma \using \rncut
        \end{prooftree}
    \]   

    For the other direction, let $\pi$ be a proof of $S, x; A,\alpha \vdash \gamma$. We obtain a proof of $S; A,\exists{x}.\alpha \vdash \gamma$ as follows.
    \[
        \begin{prooftree}
            \[
                \justifies S; A,\exists{x}.\alpha \vdash \exists{x}.\alpha \using \rnax
            \]
            \[
                \pi \leadsto S,x; A,\alpha \vdash \gamma
            \]
            \justifies S; A,\exists{x}.\alpha \vdash \gamma \using \rnexelim
        \end{prooftree}
    \]
\end{enumerate}
\end{proof}

This leads us to a notion of \emph{kernel}. 

\begin{definition}
    The \emph{atoms} of an assertion $\alpha$, denoted $\atomsof(\alpha)$, is the set of all maximal subformulas of $\alpha$ that are atomic.     
    The \emph{kernel} of $(S;A)$, denoted $\dc(S; A)$, is given by $(T;E)$ where $T = S \cup \boundvars(A)$ and $E = \{\beta \in \atomsof(\alpha) \mid \alpha \in A\}$. 
\end{definition}

Any $x \in \boundvars(A)$ which is added to $T$ can be thought of as an ``eigenvariable'' which witnesses an existential assertion in $A$. If we derive some $\gamma$ from $(T \cup \{x\}; \beta)$, since we only consider $\gamma$ such that $\varsof(\gamma)\cap\boundvars(A) = \emptyset$, we can also derive it from $(T; \exists{x}.\beta)$. Lemma~\ref{lem:leftprop} can thus always be applied, and it can be shown that kernels preserve derivability, i.e. $(S; A) \assderives \gamma$ iff $\dc(S; A) \assderives \gamma$ for any $\gamma$.

Here is another basic property of kernels, which is crucially used in many proofs later. 

\begin{lemma}\label{lem:dcpure}
    Suppose $(T;E) = \dc(S;A)$ for some $(S;A)$. 
    If $(T;E) \assderives \alpha$ and $a \in \publics(\alpha)$, then $T \DYderives a$. If $(T;E) \eqderives \equals{t}{u}$ then $T \DYderives t$ and $T \DYderives u$. 
\end{lemma}
\begin{proof}
Recall that we only consider $(S;A)$ such that $\freevars(S;A) \cap \qvar = \emptyset$, and $S \DYnderives \publics(\beta) \in S$ for all $\beta \in A$. Since $(T;E) = \dc(S;A)$, we have $T = S \cup \boundvars(A)$ and $E = \{\gamma \in \atomsof(\beta) \mid \beta \in A\}$. Thus $T \DYderives \publics(\gamma) \in T$ for every $\gamma \in E$, and $\varsof(E) \cap \qvar \subseteq T$. 

Let $\pi$ be a proof of $(T;E) \assderives \alpha$. Note that $\pi$ has no occurrence of $\rnexelim$ or $\rnconje$. We assume that all premises of $\rneq$ are normal $\DYderives$ proofs ending in a destructor (by repeatedly turning all $\text{constructor}+\rneq$ patterns into $\rneq+\rncons$). We show by induction that $T \DYderives \publics(\alpha)$. Let $\rnrule$ denote the last rule of $\pi$.
\begin{itemize}[leftmargin=*]

\item $\rnrule = \rnax$: $\alpha \in E$. So $T \DYderives \publics(\alpha)$. 

\item $\rnrule = \rneq$: $\alpha$ is $\equals{t}{t}$ with $T \DYderives t$ via a proof ending in destructor. Since any term in $T$ is either in $\qvar$ or contains no variables from $\qvar$, and since $t \in \subterms(T)$, we see that $\publics(\alpha)$ is $\{t\}$ or $\emptyset$, and $T \DYderives \publics(\alpha)$ in both cases.

\item $\rnrule \in \{\rnsymm, \rntrans, \rnprom, \rnlint, \rnsubst, \rnconji\}$: Any $t \in \publics(\alpha)$ is in $\publics(\beta)$ for one of the premises $\beta$, and the result follows.

\item $\rnrule = \rncons$: $\alpha$ is of the form $\equals{t}{u}$, where $t = \func(t_{0},t_{1})$ and $u = \func(u_{0},u_{1})$, and the immediate subproofs of $\pi$ derive $\equals{t_{0}}{u_{0}}$ and $\equals{t_{1}}{u_{1}}$. Now, any term in $\publics(\alpha)$ is a public term of one of the premises (and we can apply IH), unless it is $t$ or $u$. Say it is $t$. Then, $t$ is a maximal subterm of $\alpha$ which avoid $\qvar$, and thus it must be that $t_{0}$ and $t_{1}$ are also public terms of the premises. Thus $T \DYderives \{t_{0},t_{1}\}$ by IH, and hence $T \DYderives t$. Similarly for $u$. 

\item $\rnrule = \rnproj$: $\alpha$ is $\equals{t}{u}$, and any public term of $\alpha$ is a public term of the premise (and we can apply IH), unless it is $t$ or $u$. But by abstractability, $T \DYderives \{t,u\}$, and we are done. 

\item $\rnrule = \rnlwk$: $\alpha$ is $t \listmemb [n_{0},\ldots,n_{k}]$, where $t$ and all the $n_{i}$'s are variables or names. The premise is $\equals{t}{n_{i}}$ for some $i$, and we also require that $S \DYderives n_{i}$ for all $i$. Combining this with the IH, we see that $S \DYderives \publics(\alpha)$. 

\item $\rnrule = \rnsays$: $\alpha$ is of the form $\pk(k)\says\beta$, and $\beta$ is proved by the immediate subproof. We also have that $S \vdash k$ and hence $S \vdash \pk(k)$. Any other public term occurring in $\alpha$ occurs in $\beta$, so by IH we have that $S \DYderives \publics(\alpha)$. 

\item $\rnrule = \rnexintro$: $\alpha$ is of the form $\exists{x}.\beta$, with premise $\gamma = \replsubtermat{\beta}{P}{r}$, where $P = \posof{x}{\beta}$. We also have, by the other requirements for the rule, $T \DYderives r$ and $P \subseteq \abstractable(T\cup\{x\},\beta)$. By Lemma~\ref{lem:absderiv}, $P \subseteq \abstractable(T,\gamma)$. Consider any $a = \subtermat{\alpha}{q} \in \publics(\alpha)$. If $a\in\publics(\gamma)$, then we can apply IH. Otherwise, $q$ has to be a sibling of some position in $p \in P$. In other words, $a$ is public in $\alpha$ because its sibling is $x$, but in $\gamma$, the $x$ is replaced by $r$ (and $\varsof(r) \cap \qvar = \emptyset$), so $a$ is no longer a \emph{maximal} subterm avoiding $\qvar$. Since the set of abstractable positions is sibling-closed, $q\in\abstractable(T, \alpha)$, and since subterms at abstractable positions are derivable, $T\DYderives a$.  
\end{itemize} 

Now consider an $\eqderives$ proof of $(T;E) \vdash \equals{t}{u}$. It has been shown above that $T \DYderives \publics(\equals{t}{u})$. Consider $t$. Either $t \in \publics(\equals{t}{u})$, in which case we are done. Otherwise, every maximal subterm of $t$ which avoids $\qvar$ is derivable from $T$, and every $x \in \varsof(t) \cap \qvar$ is in $T$. From these, we can ``build up'' $t$ using constructor rules only, thereby proving that $T \DYderives t$. Similarly we can show that $T \DYderives u$. 
\end{proof}
As mentioned earlier, by proof normalization, we decompose a proof $\pi$ of $(S;A) \vdash \alpha$ into several proofs of atomic subformulas of $\alpha$ (equalities, predicates, list membership, and $\textit{says}$ assertions), and a proof $\pi_{0}$ which uses these atoms as axioms, and applies $\rnconji$ and $\rnexintro$, all with the kernel as LHS. 

For each of these atomic subformulas, we would like to operate in a proof system which does not involve conjunction or existential quantification. This is easy to do for equalities, predicates, and lists, because the only way to derive such assertions is by deriving other equalities, predicates, and lists. 

However, consider subformulas of the form $\pk(k) \says \beta$. We can derive those in two ways -- either by using $\rnax$ (if the formula is already in the LHS) or by using the $\rnsays$ rule on $\beta$ and $k$. In the latter case, $\beta$ might contain logical operators! Thus, we need to break down $\beta$ as well. 

We thus formalize the \emph{hereditary atoms} of a formula as:
\[
\eatomsof(\gamma) = \begin{cases}
\eatomsof(\alpha) \cup \eatomsof(\beta) & \text{if $\gamma = \alpha\conj\beta$} \\
\eatomsof(\alpha) & \text{if $\gamma = \exists{x}.\alpha$} \\
\left\{\pk(k)\says\alpha\right\} \cup \eatomsof(\alpha) & \text{if $\gamma = \pk(k)\says\alpha$} \\
\{\gamma\} & \text{otherwise}
\end{cases}
\]

We now reduce any proof of $S; A \assderives \alpha$ to one with a very particular structure, as depicted in Figure~\ref{fig:proofstruc}. This new proof has as its LHS the kernel $(T; E)$ of $(S; A)$, and derives $\alpha$. This proof first involves multiple proofs, each of which is an $\eqderives$ proof~\footnote{Recall that $\eqderives$ is the subsystem that does not use any rules from $\{\rnconji, \rnconje, \rnexintro, \rnexelim, \rnsays\}$.} of some hereditary atom of $\alpha$, with witnesses appropriately assigned to bound variables by a substitution $\mu$. These proofs are then followed by applications of the $\rnax$, $\rnconji$, $\rnexintro$ and $\rnsays$ rules (represented by $\vdash_{i}$ in the Figure~\ref{fig:proofstruc}) to get $\alpha$.  

\begin{figure}[!h]
\centering
\begin{tikzpicture}
\node[inner sep=0pt] (proof) at (2,-1)
{\includegraphics[width=.2\textwidth]{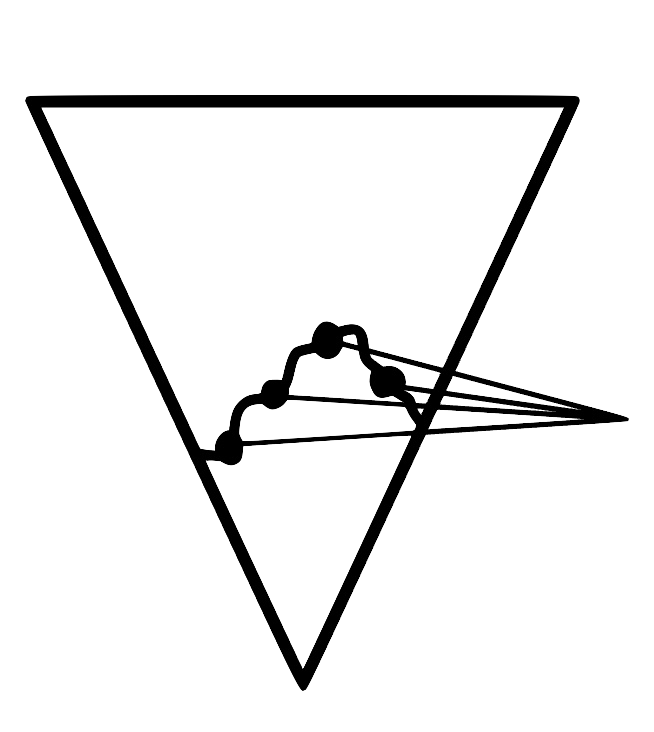}};
\node (eqd) at (2, -0.2) {$\eqderives$};
\node (id) at (2, -1.8) {$\vdash_{i}$};
\node (hat1) at (5.8, -1.3) {Each a sequent of the form};
\node (hat2) at (5.8, -1.8) {$T; E \vdash \mu(\beta)$ for $\beta \in \eatomsof(\alpha)$};
\node (conc) at (2, -3) {$T; E \vdash \alpha$};
\end{tikzpicture}
\caption{Structure of the new proof guaranteed by Theorem~\ref{thm:eqs-to-gamma}}
\label{fig:proofstruc}
\end{figure}

Consider the set $X$ of all hereditary atoms of $\alpha$ which feature in the above reduction. Suppose $\beta \in X$ is of the form $\pk(k) \says(\exists{x}.\delta)$, but $\exists{x}.\delta \notin X$. Then $\beta$  can only be derived from the LHS by the $\rnax$ rule, since there is no other rule in the $\eqderives$ system that derives a$\says$assertion. Thus we do not obtain $\exists{x}.\delta$ using the $\rnexintro$ rule, and so we do not need to provide a witness for such an $x$. This is precisely formulated in the next theorem. 

In the statement of the theorem,~\ref{item:dycond} ensures that all witnesses are derivable, \ref{item:asscond} ensures that all the atoms in $X$ have a proof (with witnesses instantiated appropriately), and \ref{item:introcond} ensures that the final intros-only proof exists. Finally, \ref{item:abscond} ensures that the proper abstractability conditions for applications of $\rnexintro$ are satisfied. For any set of assertions, we denote the set $\{x \in \boundvars(\beta) \mid \beta \in X\}$ by $\boundvars(X)$.

\begin{theorem}\label{thm:eqs-to-gamma}
For a formula $\alpha$ s.t.\ $\boundvars(\alpha) \cap \varsof(S;A) = \emptyset$, and $(T;E) = \dc(S;A)$, $(S; A) \assderives \alpha$ iff  there is $X \subseteq \eatomsof(\alpha)$ and $\mu$ with $\dom(\mu) = \boundvars(\alpha)\setminus\boundvars(X)$ s.t.: 
    \begin{enumerate}[label={[{\alph*}]}]
        \item \label{item:dycond} $\forall{}x \in \dom(\mu): T \DYderives \mu(x)$.
        \item \label{item:asscond}For all $\beta \in X$, $(T;E) \eqderives \mu(\beta)$.  
        \item \label{item:introcond} $(T;\mu(X)) \assderives \alpha$ via a proof using rules from $\{\rnax, \rnconji, \rnexintro, \rnsays\}$. 
        \item \label{item:abscond}$\forall{}x \in \dom(\mu), t \in \subterms(\alpha)$: $\posof{x}{t} \subseteq \abstractable(T \cup \dom(\mu), t)$.
    \end{enumerate}
\end{theorem}
\begin{proof}
\phantom{a}
    \begin{description}
        \item[$(\Rightarrow)$] Suppose $(S;A) \assderives \alpha$. Then, since kernels preserve derivability, $(T;E) \assderives \alpha$. Let $\pi$ be a normal proof of $(T;E) \vdash \alpha$. Since $E$ only has atomic assertions, it is easy to see that there is no occurrence of the $\rnconje$ and $\rnexelim$ rules. Recall that we only consider $\alpha$ such that no $x$ is quantified by quantifiers occurring in two distinct positions in $\alpha$, and that no variable occurs both free and bound in $\alpha$. For each $x \in \boundvars(\alpha)$ introduced in $\pi$ via an $\rnexintro$ application, let $t_{x}$ be the witness used by the $\rnexintro$ rule introducing the quantifier $\exists{x}$ in $\alpha$. Define $\mu(x) \coloneqq t_{x}$ for each such $x$. The side conditions for the $\rnexintro$ occurrences guarantee that $T \DYderives \mu(x)$ for each $x \in \dom(\mu)$, thus satisfying \ref{item:dycond}. 
        
        Let $X \subseteq \eatomsof(\alpha)$ be all the hereditary atoms of $\alpha$ appearing on the RHS in various subproofs of $\pi$. By normalization, one can always place the logical rules after deriving atomic formulas. Hence, we can decompose $\pi$ into proofs $\pi_{\beta}$ of $(T;E) \vdash \mu(\beta)$ for each $\beta \in X$, and a proof $\wh{\pi}$ deriving $(T;\mu(X)) \vdash \alpha$ using only the $\rnax, \rnconji, \rnexintro$ and $\rnsays$ rules. This proves \ref{item:asscond} and \ref{item:introcond}. 
        
        We now prove \ref{item:abscond}. It is evident that each subproof of $\wh{\pi}$ has conclusion $\mu(\beta)$ for some $\beta \in \subformulas(\alpha)$, with $\wh{\pi}$ itself deriving $\mu(\alpha) = \alpha$. We will now show that for every subproof $\pi_{0}$ of $\wh{\pi}$ with conclusion $\mu(\beta)$ and last rule $\rnrule$, we have (letting $Z_{\beta} = \boundvars(\beta)\setminus\boundvars(X)$): 
        \begin{equation} \label{eqn:abscond}
        \forall{}x \in Z_{\beta}, \forall{}t \in \subterms(\mu(\beta)): \posof{x}{t} \subseteq \abstractable(T \cup Z_{\beta}, t).
        \end{equation}
       \begin{description}
            \item[$\rnrule = \rnax$:] $\mu(\beta) \in \mu(X)$, so $Z_{\beta} = \emptyset$, and so (\ref{eqn:abscond}) holds vacuously. 
            \item[$\rnrule = \rnconji$:] $\beta$ of the form $\beta_{0} \conj \beta_{1}$, and $\boundvars(\beta_{0})$ and $\boundvars(\beta_{1})$ are disjoint, and no variable has both free and bound occurrences. So no variable in $\boundvars(\beta_{i})$ occurs in $\beta_{1-i}$. So if $x \in \boundvars(\beta_{i})$, and any $t \in \subterms(\mu(\beta_{1-i}))$, then $\posof{x}{t} = \emptyset$. So~(\ref{eqn:abscond}) for $\pi_{0}$ follows by IH (applied on the immediate subproofs). 
            	\item[$\rnrule = \rnsays$:] $\beta$ is of the form $\pk(k)\says\beta'$ and every bound variable of $\beta$ is also bound in $\beta'$, so we get~(\ref{eqn:abscond}) from IH. 
            \item[$\rnrule = \rnexintro$:] $\beta = \exists{z}.\gamma$, and $\mu(\beta) = \exists{z}.\mu'(\gamma)$, where $\mu' = \mu \restr (Z_{\gamma})$. The immediate subproof of $\pi_{0}$ has conclusion $\mu(\gamma)$.  
            
            Now for any $r \in \subterms(\mu(\beta))$, letting $P = \posof{z}{r}$, $t = \replsubtermat{r}{P}{\mu(z)} \in \subterms(\mu(\gamma))$. For any $x \in \vars$, we have $\posof{x}{r} = \posof{x}{t} \cap \positions{r}$ and $\abstractable(T\cup Z_{\beta},r) = \abstractable(T\cup Z_{\gamma}, t) \cap \positions{r}$ (by Lemma~\ref{lem:absderiv}).
            
            By IH, for all $x \in Z_{\gamma}$ and $t \in \subterms(\mu(\gamma))$, $\posof{x}{t} \subseteq \abstractable(T \cup Z_{\gamma}, t) \subseteq \abstractable(T \cup Z_{\beta}, t)$. So for all $x \in Z_{\beta}\setminus\{z\}$ and $r \in \subterms(\mu(\beta))$, $\posof{x}{r} \subseteq \abstractable(T \cup Z_{\beta}, r)$.   
            
            For $z$, the abstractability side condition for $\rnexintro$ implies that for all $r \in \subterms(\mu(\beta))$, $\posof{z}{r} \subseteq \abstractable(T \cup Z_{\beta}, r)$. Thus, equation~(\ref{eqn:abscond}) follows for $\pi_{0}$.         
        \end{description}
        
        Applying~(\ref{eqn:abscond}) to $\wh{\pi}$, we get \ref{item:abscond}. 
        \item[$(\Leftarrow)$] This is the easier direction. We just compose all the $\eqderives$ proofs and the intros-only proof to obtain an $\assderives$ proof $\pi$ of $(T;E) \vdash \alpha$. The abstractability condition~\ref{item:abscond} ensures that the $\rnexintro$ is always enabled in $\pi$. 
           \end{description}
                  Thus, $(T;E) \assderives \alpha$ iff $(S;A) \assderives \alpha$, and so we are done. \qedhere   

\end{proof}

For the rest of the paper, we use the following notation. $(T_{i};E_{i}) \coloneqq \dc(\knowfunc_{i}(I))$ and $(U_{i};F_{i}) \coloneqq \dc(\knowfunc_{i}(u_{i}))$ for $1 \le i \le n$. Note that $T_{i} \subseteq T_{i+1}$ and $E_{i} \subseteq E_{i+1}$ for every $i$.

Since $\dom(\intsub) = \freevars(\xi)$, we have $\intsub(x) = x$ for all $x \in \qvar$. It follows that $\intsub(\dc(S;A)) = \dc(\intsub(S;A))$, for any $(S;A)$. 

Applying Theorem~\ref{thm:eqs-to-gamma} to the $\intsub(\knowfunc_{i-1}(I)) \assderives \intsub(\beta_{i})$ derivations in Definition~\ref{def:protruns}, for every $i \leq n$ we get $X_{i} \subseteq \eatomsof(\beta_{i})$ and a substitution $\iwsub_{i}$ with domain $\boundvars(\beta_{i})\setminus\boundvars(X_{i})$ s.t.: 
\begin{itemize}
    \item for every $x \in \dom(\iwsub_{i})$, $\intsub(T_{i-1}) \DYderives \iwsub_{i}(x)$, and
        \item $\intsub(T_{i-1}; E_{i-1}) \eqderives \intsub\iwsub_{i}(\gamma)$ for $\gamma \in X_{i}$.
\end{itemize}

For every $i \leq n$, Definition~\ref{def:protruns} also states $\knowfunc_{i}(u_{i}) \assderives \alpha_{i}$, and thus, $\intsub(\knowfunc_{i}(u_{i})) \assderives \intsub(\alpha_{i})$. So Theorem~\ref{thm:eqs-to-gamma} guarantees $Y_{i} \subseteq \eatomsof(\alpha_{i})$ and a substitution $\hwsub_{i}$ with domain $\boundvars(\alpha_{i})\setminus\boundvars(Y_{i})$ s.t.: 
\begin{itemize}
\item for every $x \in \dom(\hwsub_{i})$, $\intsub(U_{i}) \DYderives \hwsub_{i}(x)$, and
 \item $\intsub(U_{i}; F_{i}) \eqderives \intsub\hwsub_{i}(\gamma)$, where $\gamma \in Y_{i}$. 
\end{itemize}

For any $\gamma \in X_{i} \cup Y_{i}$, three possibilities arise. 
\begin{itemize}
\item $\gamma$ is of the form $\equals{t}{u}$.
\item $\gamma$ is of the form $\pk(k)\says\delta$. Such a formula can only be derived using $\rnax$, as no other rule in the $\eqderives$ system generates it. Hence such assertions can be ignored for the rest of this section, which is about preserving non-trivial $\eqderives$ proofs even after changing some substitutions. 
\item $\gamma$ is of the form $P(u_{0},\ldots,u_{m})$ or $t \listmemb \ell$. Such formulas only mention variables or names, so $\lambda(x)$ is already small for any $\lambda \in \{\intsub, \hwsub_{i}, \iwsub_{i} \mid i \leq n\}$ and any variable $x$ occurring in $\gamma$. Hence we can ignore such formulas too for the rest of the section, since these formulas do not undergo any change.
\end{itemize}

Hence we simplify the presentation for the rest of this section by only considering equality assertions $\gamma$. 

We now have, for every $i \le n$, substitutions $\iwsub_{i}$ and $\hwsub_{j}$, each with domain $\boundvars(\beta_{i})$ and $\boundvars(\alpha_{j})$. However, these substitutions do not necessarily map variables to ground terms. It is possible that $\hwsub_{j}(\alpha_{j})$ has as a subterm a variable from the domain of some ``earlier'' $\iwsub_{i}$, i.e. one where $i < j$.  

If $(T; E) \vdash \equals{x}{y}$, then $x$ and $y$ ought to actually stand for the same ground term. To capture this, we need a ``compound'' substitution that maps each variable in the domain of each $\iwsub$ and each $\hwsub$ to a ground term. We now present a motivating example which is followed by the formal definition of this ground substitution.

\begin{example}
    Suppose $y \in \boundvars(\beta_{1})$, and $x \in \boundvars(\alpha_{2})$. Consider a situation where $\hwsub_{2}(x) = \{y\}_{k}$ and $\iwsub_{1}(y) = (m_{0}, m_{1})$, where $m_{0}, m_{1} \in \names$. Also suppose $(T_{2};E_{2}) \vdash \equals{x}{z}$ for some $z \in \dom(\intsub)$. We need a $\lambda$ which maps $x$ and $z$ to the same ground term, i.e. $\lambda$ needs to be s.t. $\lambda(x) = \lambda(z)$. We can take $\lambda$ to be $\intsub\iwsub_{1}\hwsub_{2}$. We see that $\lambda(x) = \intsub(\iwsub_{1}(\hwsub_{2}(x))) = \intsub(\iwsub_{1}(\{y\}_{k})) = \intsub(\{(m_{0},m_{1})\}_{k}) = \{(m_{0},m_{1})\}_{k}$. Observe that $\dom(\lambda) = \dom(\intsub) \cup \dom(\iwsub_{1}) \cup \dom(\hwsub_{2})$, and since $z \notin \dom(\iwsub_{1}) \cup \dom(\hwsub_{2})$, $\lambda(z) = \intsub(z)$. 
\end{example}

\begin{definition}\label{def:bigsub}
The compound substitution which maps any variable in $\dom(\intsub) \cup \{\dom(\iwsub_{i}) \cup \dom(\hwsub_{i}) \mid 1 \le i \le n\}$ to a ground term is given by $\bigsub \coloneqq \intsub\iwsub_{1}\hwsub_{1}\ldots\iwsub_{n}\hwsub_{n}$. 
\end{definition}

Note that for $\lambda \in \{\intsub, \hwsub_{i}, \iwsub_{i} \mid i \leq n\}, \bigsub(\lambda(x)) = \bigsub(x)$.

\begin{lemma}\label{lem:equndertheta}
	Suppose $\lambda$ is such that $\lambda(r) = \lambda(s)$ for each $\equals{r}{s} \in E$, and $T; E \eqderives \equals{t}{u}$. Then $\lambda(t) = \lambda(u)$.
\end{lemma}
\begin{proof}
	Suppose $T; E \derives \equals{t}{u}$ via a proof $\pi$ with last rule $\rnrule$. The proof is by induction on the structure of $\pi$. The following cases arise. 
	\begin{itemize}
		\item $\rnrule = \rnax$: In this case, $\equals{t}{u} \in E$, so by assumption, $\lambda(t) = \lambda(u)$.
		\item $\rnrule = \rneq$: In this case $t = u$, so $\lambda(t) = \lambda(u)$ as well.
		\item $\rnrule = \rntrans$: Suppose $\equals{t_{0}}{t_{1}}, \ldots, \equals{t_{n-1}}{t_{n}}$ are the premises of $\rnrule$, with $t = t_{0}$ and $u = t_{n}$. By IH, $\lambda(t_{i-1}) =  \lambda(t_{i})$ for all $i \leq n$. It follows that $\lambda(t) = \lambda(u)$.
		\item $\rnrule = \rncons$: Let $t = \func(t_{1}, \ldots, t_{n})$ and $u = \func(u_{1}, \ldots, u_{n})$ and let $\equals{t_{1}}{u_{1}}, \ldots, \equals{t_{n}}{u_{n}}$ be the premises of $\rnrule$. By IH, $\lambda(t_{i}) = \lambda(u_{i})$ for all $i \leq n$. Thus we have the following:
		\begin{tabbing}
		    \hspace{0.25cm}\=
            $\lambda(t)$ \=$= \lambda(\func(t_{1}, \ldots, t_{n})) = \func(\lambda(t_{1}), \ldots, \lambda(t_{n}))$ \\ \>\>
			$=\func(\lambda(u_{1}), \ldots, \lambda(u_{n})) = \lambda(\func(t_{1}, \ldots, t_{n})) = \lambda(u)$.
		\end{tabbing}
		\item $\rnrule = \rnproj$: Let $\equals{\func(t_{1}, \ldots, t_{n})}{\func(u_{1}, \ldots, u_{n})}$ be the premise of the last rule with $t = t_{i}$ and $u = u_{i}$ respectively. By IH, $\lambda(\func(t_{1}, \ldots, t_{n})) = \lambda(\func(u_{1}, u_{n}))$. So, $\lambda(t) = \lambda(u)$. 
	\end{itemize}
\end{proof}

\begin{lemma}\label{lem:prooftotruth}
	For any $i \in \{1, \ldots, n\}$, 
	\begin{enumerate}[leftmargin=*]
		\item if $\equals{t}{u} \in E_{i} \cup F_{i}$, then $\bigsub(t) = \bigsub(u)$.
		\item if $\intsub(T_{i-1}; E_{i-1}) \eqderives \intsub\iwsub_{i}(\equals{t}{u})$, then $\bigsub(t) = \bigsub(u)$.
		\item if $\intsub(U_{i}; F_{i}) \eqderives \intsub\hwsub_{i}(\equals{t}{u})$, then $\bigsub(t) = \bigsub(u)$.
	\end{enumerate}
\end{lemma}
\begin{proof}
	In addition to $E_{i}, F_{i}$ for $0 < i \leq n$, we also use $E_{0} = \emptyset$, for which claim 1 is vacuously true. We prove the claims simultaneously by induction on $i > 0$. Assume that they hold for all $j < i$ via IH1, IH2, and IH3. 

	\begin{enumerate}
		\item Suppose $\equals{t}{u} \in E_{i}$. Then, $\exists{}j < i: \equals{t}{u} \in \SF(\alpha_{j})$, and $\intsub(U_{j};F_{j}) \eqderives \intsub\hwsub_{j}(\equals{t}{u})$. By IH3, $\bigsub(t) = \bigsub(u)$. If $\equals{t}{u} \in F_{i}$, then $\exists{}j \leq i: \equals{t}{u} \in \inteqs_{j}$, and $\intsub(T_{j-1};E_{j-1}) \eqderives \intsub\iwsub_{j}(\equals{t}{u})$. If $j < i$, by IH2, $\bigsub(t) = \bigsub(u)$. If $j = i$, by IH1, $\bigsub({r}) = \bigsub({s})$ for every $\equals{{r}}{{s}} \in E_{i-1}$. Any $\equals{a}{b} \in \intsub(E_{i-1})$ is of the form $\intsub(\equals{r}{s})$ for some $\equals{r}{s} \in E_{i-1}$. Thus, $\bigsub(a) = \bigsub(\intsub(r)) = \bigsub(r) = \bigsub(s) = \bigsub(\intsub(s)) = \bigsub(b)$. By Lemma~\ref{lem:equndertheta}, $\bigsub(\intsub\iwsub_{j}(t)) = \bigsub(\intsub\iwsub_{j}(u))$, i.e. $\bigsub(t) = \bigsub(u)$. 
		
        \item Suppose $\intsub(T_{i-1}); \intsub(E_{i-1}) \eqderives \intsub\iwsub_{i}(\equals{t}{u}$. As above, for each $\equals{a}{b} \in \intsub(E_{i-1})$, $\bigsub(a) = \bigsub(\wh{s})$. By appealing to Lemma~\ref{lem:equndertheta}, we get $\bigsub(\intsub\iwsub_{i}(t)) = \bigsub(\intsub\iwsub_{i}(u))$, i.e. $\bigsub(t) = \bigsub(u)$. 
        
        \item The proof is similar to the above. \qedhere
	\end{enumerate}
\end{proof}

We developed this preliminary setup for both honest agent derivations as well as intruder derivations in order to demonstrate the interplay between $\hwsub$ and $\iwsub$, as evidenced in the definition of $\bigsub$. However, the insecurity problem itself is concerned only with intruder derivability, and therefore, in the next few sections we will focus only on $\beta_{i}, (T_{i}; E_{i})$, and $\iwsub_{i}$. We will discuss honest agent derivations later.

\subsection{Typed proofs for \texorpdfstring{$\DYderives$}{Dolev-Yao derivability} and \texorpdfstring{$\eqderives$}{derivability of equalities}}

In order to obtain ``small'' versions of the various substitutions $\intsub, \hwsub_{i},$ and $\iwsub_{i}$ while preserving their interaction, we consider a universe of ``anchor terms''. These are abstract terms that appear in the protocol specification, and for which we have a bound on size. We call these anchors ``types''. We would eventually like to be able to convert any proof into one that only involves typed terms, i.e. terms that correspond to one of these types under $\bigsub$. 

\begin{definition}[Types and typed terms]\label{def:types}
We use the sets $\constst$ (consisting of the terms occurring in $\xi$ before applying any substitution) and $\stnonvars$ (the same set, but without variables) to \defemph{type} the terms appearing in any proof. \[ \constst \coloneqq \bigcup_{i \leq n}\bigl\{\bigl(\subterms(T_{i} \cup U_{i}) \cup \subterms(E_{i} \cup F_{i})\bigr)\bigr\} \quad\quad \stnonvars \coloneqq \constst \setminus \vars \]
A term $t$ is \defemph{typed} if $t \in \intsub(\stnonvars) \cup \bigsub(\constst) \cup \qvar$. 
\end{definition}

Note that we must consider $\intsub(\stnonvars)$ separately from $\bigsub(\constst)$. Consider a term of the form $(m, x) \in \stnonvars$, where $x \notin \dom(\intsub)$. $\intsub((m, x)) = (m, x)$, but this cannot be in $\bigsub(\constst)$, since $\bigsub(\constst)$ only contains ground terms. Thus, $\intsub(\stnonvars) \not\subseteq \bigsub(\constst)$.

We now define a notion of ``zappable terms'', which are terms that do not correspond to any type in $\constst$. The idea is these terms can be freely ``zapped''.\footnote{In order to motivate the key ideas behind typing, we will often use the word ``zap'' to mean replacing terms by an atomic name. However, we will formally define this zapping operation in the next subsection.}
\begin{definition}[Zappable terms]\label{def:zappable}
    A term $t$ is \defemph{zappable} if there is an $x \in \dom(\bigsub)$ such that $\bigsub(t) = \bigsub(x)$, but there is no $u \in \stnonvars$ such that $\bigsub(x) = \bigsub(u)$. We refer to such an $x$ as a \defemph{minimal variable}.
\end{definition}
Here are a couple of easy observations that relate to zappable terms. 

\begin{obs}\label{obs:tzap-iff-uzap}
\phantom{a}
\begin{itemize}
\item If a term $t$ is zappable, then $t \notin \stnonvars$.
\item If a term $t \in \bigsub(\constst)$ is not zappable, then $t \in \bigsub(\stnonvars)$. 
\item For $t, u$ s.t.\ $\bigsub(t) = \bigsub(u)$, $t$ is zappable iff $u$ is zappable. 
\end{itemize}
\end{obs}

\begin{lemma}\label{lem:projcase}
    Suppose $t = \func(t_{0},t_{1})$ and $u = \func(u_{0},u_{1})$ are typed, and $\bigsub(t) = \bigsub(u)$. One of the following is true:
 \begin{itemize}[leftmargin=*]
        \item $t$ and $u$ are not zappable, and $t_{0}, t_{1}, u_{0}, u_{1}$ are typed, or
        \item $t$ and $u$ are zappable, and $t = u$.  
    \end{itemize}   
\end{lemma}
\begin{proof} 
    Observe that for any $a \in \intsub(\constst)$, $a \in \intsub(\stnonvars)$, or $a = \intsub(x) = x$ for some $x \notin \dom(\intsub)$ (in which case $a \in \qvar$), or $a = \intsub(x)$ for $x \in \dom(\intsub)$ (in which case $a = \bigsub(x)$ also, so $a \in \bigsub(\constst)$). Thus $\intsub(\constst) \subseteq \intsub(\stnonvars) \cup \bigsub(\constst) \cup \qvar$. 

    Now $t$ and $u$ are typed, and are non-atomic. So $t, u \notin \qvar$, and so $t, u \in \intsub(\stnonvars) \cup \bigsub(\constst)$. We consider two cases: 
    \begin{itemize}[leftmargin=*]
        \item {\bf Neither $t$ nor $u$ is zappable:} Consider $t$. If $t \in \intsub(\stnonvars)$, each $t_{i} \in \intsub(\constst) \subseteq \intsub(\stnonvars) \cup \bigsub(\constst) \cup \qvar$. If $t \in \bigsub(\constst)$, then since $t$ is not zappable, $t = \bigsub(a)$ for some $a \in \stnonvars$. Then $a$ has to be of the form $\func(a_{1},\dots,a_{k})$, with each $a_{i} \in \constst$ and $t_{i} = \bigsub(a_{i})$. Thus each $t_{i} \in \bigsub(\constst) \subseteq \intsub(\stnonvars) \cup \bigsub(\constst) \cup \qvar$. Reasoning about $u$ in a similar manner, we see that each $u_{i} \in \intsub(\stnonvars) \cup \bigsub(\constst) \cup \qvar$. So each $t_{i}$ and $u_{i}$ is typed.
        \item {\bf One of $t$ and $u$ is zappable:} Say $t$ is zappable. Then, since $\bigsub(t) = \bigsub(u)$, $u$ is zappable as well. Therefore $t, u \notin \intsub(\stnonvars)$, which implies that $t, u \in \bigsub(\constst)$. Therefore both $t$ and $u$ are ground terms, so $t = \bigsub(t) = \bigsub(u) = u$. \qedhere
    \end{itemize}
\end{proof}

We now devise notions of ``typed proofs'' for the $\DYderives$ as well as the $\eqderives$ system, which will help us obtain bounds on the sizes of terms appearing in the ranges of various substitutions. Then, we show that every proof in these systems can be converted into a typed proof.

Consider a proof $\pi$ witnessing $\intsub(T_{i}) \DYderives t$ for some $t$. Any term in $T_{i}$, since $T_{i}$ is part of a kernel, is either a bound variable outside the domain of $\intsub$ (i.e. in $\qvar$) or a public term of some assertion. Note that any variables in public terms of assertions must not be quantified, hence they fall into the domain of $\intsub$. Thus, any such $t$ derived from $\intsub(T_{i})$ is either in $\qvar$, or a ground term of the form $\intsub(v)$ for some $v$.

Now, it is possible that $\pi$ mentions some term $u \not\in \bigsub(\constst)$, even if $t \in \bigsub(\constst)$. If a destructor rule is applied to $u$ in order to obtain a proof of $t$, we cannot ``zap'' $u$ into an atomic name while still preserving derivability. This leads us to the following definition of a typed proof in the $\DYderives$ system, which preserves derivability even after zapping variables as necessary. 

\begin{definition}\label{def:welltyped}[Typed $\DYderives$ proof]
	A $\DYderives$ proof $\pi$ is \defemph{typed} if for each subproof $\pi'$, either $\pi'$ ends in a constructor rule, or $\concof(\pi') \in \intsub(\stnonvars) \cup \qvar$, where $\concof(\pi')$ denotes the conclusion derived using $\pi$. 
\end{definition}

Armed with this definition of a typed $\DYderives$ proof, we can show that any proof $\intsub(T_{i}) \DYderives t$ can be transformed into a typed normal equivalent witnessing the same. 
This transformation crucially uses the following fact about how non-typed terms are generated: any non-typed term $u$ occurring in a $\DYderives$ proof from $\intsub(T_{i})$ obeys the following:
\begin{itemize}
\item appears first as part of a received assertion $\intsub(\beta)$, and
\item is generated by the intruder by putting information together, i.e. via a normal proof ending in a constructor. 
\end{itemize} 

The intuition behind this is easy to see -- honest agents follow the protocol, and will only communicate terms that follow the protocol specification, modulo any insertions by the intruder. Terms that correspond to ones in the protocol specification are always typed, so any non-typed term must have been initially sent out by the intruder, i.e. in a $\beta$ received by an honest agent. In particular, such a term must have been constructed by the intruder by putting information together, since up till that point, the intruder's knowledge state would have only consisted of typed terms, and destructor rules would preserve ``typability''. Thus, for any non-typed term $t$ such that $t \in \subterms(\intsub(T_{i}))$, we can always ``chase back'' to an index $j < i$ at which it was \emph{not} in the subterms of $\intsub(T_{j})$, but still derivable, i.e. $\intsub(T_{j}) \DYderives t$ via a normal proof ending in a constructor rule. This reasoning closely follows the ideas in~\cite{RT03}, and is formalized below. 

\begin{obs}\label{obs:earlierbetaj}
Since agent variables are mapped to names, the only free variables in sessions are intruder variables. Thus, for any $i \leq \ell$ and any $x \in \freevars(\alpha_{i})$, there is $j < i$ s.t.\ $x \in \freevars(\beta_{j})$.
\end{obs}

We define $\intterms_{i} \coloneqq \publics(\beta_{i})$ and $\honterms_{i} \coloneqq \publics(\alpha_{i})$.~\footnote{These stand for intruder terms and honest agent terms respectively.}

\begin{restatable}{lemma}{sigmaxinsigmaTp}\label{lem:sigmax-in-sigmaTp}
    Suppose $t \notin \intsub(\stnonvars) \cup \qvar$. For any $i \leq n$, if $t \in \subterms(\intsub(T_{i}))$, then there is a $k < i$ such that $t \in \subterms(\intsub(\intterms_{k}))$.
\end{restatable}
\begin{proof}
    Consider $t \in \subterms(\intsub(u)) \setminus (\intsub(\stnonvars) \cup \qvar)$ for some $u \in T_{i}$. Then, $t \in \subterms(\intsub(y))$ for some $y \in \varsof(u)$. Since $u \in T_{i}$, there is a $j < i$ such that $u \in \honterms_{j} \cup \qvar$. If $u \in \qvar$, then $u = y = \sigma(y)$ and $t = y$, but we know that $t \not\in \qvar$. Thus $u \not\in \qvar$ and $u \in \honterms_{j}$, i.e. $y \in \varsof(\honterms_{j})$. Now $\xi$ is an interleaving of sessions of $\prot$, and $y \in \varsof(u)$ where $u$ occurs in an honest agent send in a session. Thus by Observation~\ref{obs:earlierbetaj}, there is an earlier intruder send in the same session in which $y$ occurs. This send occurs before $\alpha_{j}$ in $\xi$. Thus there is a $k \leq j$ such that $y \in \varsof(\publics(\beta_{k})) = \varsof(\intterms_{k})$. Thus, $t \in \subterms(\intsub(\intterms_{k}))$.
\end{proof}

\begin{restatable}{lemma}{earlierproof}\label{lem:earlier-proof}
    Suppose $i \leq n$, $t \notin \intsub(\stnonvars) \cup \qvar$ and $\intsub(T_{i}) \DYderives t$ via a normal proof $\pi$ ending in a destructor rule. Then there is an $\ell < i$ such that $\intsub(T_{\ell}) \DYderives t$.
\end{restatable}
\begin{proof}
  	Since $\pi$ ends in a destructor rule, $t \in \subterms(\intsub(T_{i}))$. By Lemma~\ref{lem:sigmax-in-sigmaTp}, there is an $i' < i$ such that $t \in \subterms(\intsub(\intterms_{i'}))$. 
    Let $j$ be the earliest such index, and let $a \in \intterms_{j}$ such that $t \in \subterms(\intsub(a))$. Since $\intsub(T_{j-1};E_{j-1}) \assderives \intsub\iwsub_{j}(\beta_{j})$, and $a \in \intterms_{j} = \publics(\beta_{j})$, it follows by Lemma~\ref{lem:dcpure} that $\intsub(T_{j-1}) \DYderives \intsub\iwsub_{j}(a)$. But $\varsof(a) \cap \dom(\iwsub_{j}) = \emptyset$, so $\intsub(T_{j-1}) \DYderives \intsub(a)$, via a normal proof $\rho$. Consider a minimal subproof $\chi$ of $\rho$ such that $t \in \subterms(\concof(\chi))$. (There is at least one such subproof, namely $\rho$.) If $\chi$ ends in a destructor, then $\concof(\chi) \in \subterms(\intsub(T_{j-1}))$, and hence $t \in \subterms(\intsub(T_{j-1}))$. But by Lemma~\ref{lem:sigmax-in-sigmaTp}, there must be a $k < j-1$ such that $t \in \subterms(\intsub(\intterms_{k}))$, contradicting the fact that $j$ is the earliest such index. So $\chi$ ends in a constructor rule. If $t \neq \concof(\chi)$, then $t \in \subterms(\concof(\chi'))$, for some proper subproof of $\chi$. But this cannot be, since $\chi$ is a minimal proof with this property. Thus, $t = \concof(\chi)$ and $\chi$ is a proof of $\intsub(T_{j-1}) \vdash t$ (and we choose our $\ell$ to be $j-1$).           
\end{proof}

\begin{theorem}\label{thm:rustur}
	For all $t$ and all $i \in \{0,\ldots,n\}$, if $\intsub(T_{i}) \DYderives t$, then there is a typed normal proof $\pi^{*}$ of the same.
\end{theorem}
\begin{proof}
    Assume the theorem holds for all $j < i$. We show how to transform any proof $\pi$ of $\intsub(T_{i}) \vdash t$ ending in rule $\rnrule$ into a typed normal proof $\pi^{*}$ of the same by induction on the structure of $\pi$. 
    \begin{itemize}[leftmargin=*]
        \item {$r$ is $\textsf{ax}$:} $t \in \intsub(T_{i}) \subseteq \intsub(\constst)$. 
        If $t \in \intsub(\stnonvars) \cup \qvar$, we take $\pi^{*}$ to be $\pi$ itself. Otherwise, there is a $j < i$ such that $\intsub(T_{j}) \DYderives t$. We can get a typed normal proof $\pi^{*}$ of $\intsub(T_{j}) \vdash t$ and obtain the required result by weakening the LHS.

        \item {$r$ is a constructor:} We can find typed normal equivalents for all immediate subproofs, and apply the same constructor rule to get the desired $\pi^{*}$. 
        
        \item {$r$ is a destructor:} Let $\pi_{1}$ and $\pi_{2}$ be immediate subproofs of $\pi$, with $\concof(\pi_{1}) = s$, and $t$ an immediate subterm of $s$. We can find typed normal equivalents $\pi^{*}_{1}$ and $\pi^{*}_{2}$. If $\pi^{*}_{1}$ ends in a constructor, then we choose $\pi^{*}$ to be the immediate subproof of $\pi^{*}_{1}$ s.t.\ $\concof(\pi^{*}) = t$. 
        
        If $\pi^{*}_{1}$ does not end in a constructor, $s \in \intsub(\stnonvars) \cup \qvar$. Since a destructor rule $\rnrule$ was applied on $s$, $s \notin \qvar$. So $s \in \intsub(\stnonvars)$, and hence $t \in \intsub(\constst)$. If $t \in \intsub(\stnonvars) \cup \qvar$, we obtain a typed normal $\pi^{*}$ by applying $\rnrule$ on $\pi^{*}_{1}$. Otherwise, as with $\rnax$, we get a typed and normal proof $\pi^{*}$ of $\intsub(T_{j}) \vdash t$ for some $j < i$ and apply weakening. \qedhere 
    \end{itemize}
\end{proof}

Having shown that we can always obtain a typed $\DYderives$ proof, we now consider $\eqderives$. We present below an example which will motivate our choices for the definition of a \emph{typed $\eqderives$ proof}. 

Suppose $\intsub(x) = (t_{1}, t_{2})$ for some minimal $x$, and $\intsub(u) = (u_{1}, u_{2})$ for some term $u$. Suppose we also have a proof of $\equals{t_{1}}{u_{1}}$ obtained by applying $\rnproj_{1}$ to a proof of $\equals{\intsub(x)}{\intsub(u)}$, and we want a ``corresponding'' proof, even after zapping. However, $x$ would be zapped to a name, and we cannot apply $\rnproj$ to an atomic value. We would prefer a proof which allows us to preserve its structure even after zapping. To this end, we define a typed $\eqderives$ proof as follows.

\begin{definition}\label{def:welltypedeq}[Typed $\eqderives$ proof]
    A proof $\pi$ of $X; A \vdash \equals{r}{s}$ is \defemph{typed} if for every subproof $\pi'$ with conclusion $X; A \vdash \equals{t}{u}$,
    \begin{itemize}
        \item $\pi'$ contains an occurrence of the $\rncons$ rule, or
        \item $t = u$, or 
        \item $t$ and $u$ are typed terms.
    \end{itemize}
\end{definition}

Intuitively, this definition disallows ``asymmetric'' zapping of the above kind, and allows us to prove the equivalent of Theorem~\ref{thm:rustur} for $\eqderives$ proofs. 

\begin{restatable}{theorem}{rustureq}\label{thm:rustur-eq}
    For $i \leq n$ and $a,b \in \Tterms$, if $\intsub(T_{i}; E_{i}) \eqderives \equals{a}{b}$, then there is a typed normal proof of $\intsub(T_{i}; E_{i}) \vdash \equals{a}{b}$.
\end{restatable}

\begin{proof}[Proof of Theorem~\ref{thm:rustur-eq}]
By Theorem~\ref{thm:eqnormsub}, we know that every $\eqderives$ proof can be converted to an equivalent normal proof. We can show that every normal $\eqderives$ proof is typed. The only non-trivial case is when the last rule is $\rnproj$. Consider a normal proof $\pi$ of $\intsub(T_{i};E_{i}) \vdash \equals{a}{b}$, whose last rule is $\rnproj$, and whose immediate (typed normal, by IH) subproof is $\pi'$ deriving $\equals{\func(a,c)}{\func(b,d)}$. 
Since $\pi$ is a normal proof ending in $\rnproj$, the $\rncons$ rule does not occur in $\pi$ or $\pi'$. Two cases arise: 
\begin{itemize}
\item $\func(a,c) = \func(b,d)$, in which case $a = b$ and $\pi$ is typed.
\item $\func(a,c)$ and $\func(b,d)$ are both typed terms. By Lemma~\ref{lem:projcase}, either $\func(a,c) = \func(b,d)$ (whence $a = b$), or $a, b, c, d$ are all typed, and thus $\pi$ is typed. \qedhere
\end{itemize}
\end{proof}

\subsection{Small substitutions \texorpdfstring{$\vintsub, \vbigsub$, and $\viwsub_{i}$}{}}

Assume that there is an $\fixedname \in T_{0} \cap \names$ s.t.\ $\fixedname \notin \subterms(\{\alpha_{i}, \beta_{i}\}) \cup \subterms(\rng(\hwsub_{i}) \cup \rng(\iwsub_{i}))$ for all $i$. This can be thought of as a fixed ``spare name'' that does not appear in the run.  We will use this name to formally define a zap operation, as below.

\begin{definition}\label{def:zap}
    For any term $t$, we inductively define the \defemph{zap} of $t$, denoted $\zap{t}$, as follows: 
    \begin{align*}
        \zap{x} &\coloneqq x \\ 
        \zap{n} &\coloneqq \begin{cases}
        	\fixedname & \qquad \hspace{2mm} \mbox{if $n$ is zappable} \\
            	n & \qquad \hspace{2mm} \mbox{otherwise}
        \end{cases} \\
        \zap{\func(t_{1}, t_{2})} &\coloneqq \begin{cases}
            \fixedname & \mbox{if $\func(t_{1}, t_{2})$ is zappable} \\
            \func(\zap{t_{1}}, \zap{t_{2}}) & \mbox{otherwise}
        \end{cases}
    \end{align*}
    For a set of terms $X$, $\zap{X} \coloneqq \{\zap{t} \mid t \in X\}$. For a set of equalities $E$, $\zap{E} \coloneqq \{\equals{\zap{t}}{\zap{u}} \mid \equals{t}{u} \in E\}$.
\end{definition}

\begin{definition}\label{def:vlambda}
    For $\lambda \in \{\intsub, \bigsub, \iwsub_{i} \mid i \leq n\}$, the \defemph{small substitution} $\vlambda$ corresponding to $\lambda$ is defined as $\vlambda(x) \coloneqq \zap{\lambda(x)}$ for all $x \in \vars$. 
\end{definition}

Here are a few examples that illustrate the above definition, for different choices of $\lambda$ and $\constst$. 
\begin{example}\label{ex:zap}
    \phantom{a}
    \begin{enumerate}
        \item 
        Suppose $\constst = \subterms(\{\fixedname, y, (y_{1}, \{y_{2}\}_{k}) \})$, where $y_{1}, y_{2}$ are minimal, and $\iwsub_{2}(y) = (y_{1}, \{y_{2}\}_{k})$. Then $\viwsub_{2}(y) = (y_{1}, \{y_{2}\}_{k})$ and $\vbigsub(y) = (\fixedname, \{\fixedname\}_{k})$.
        \item 
        Suppose $\constst = \subterms(\{\fixedname, y, y_{2}, (y_{1}, x)\})$ and $\iwsub_{2}$ is the same as above, with $x$ minimal and $\intsub(x) = \iwsub_{2}(\{y_{2}\}_{k})$. Then $\viwsub_{2}(y) = (y_{1}, \fixedname)$ and $\vbigsub(y) = (\fixedname, \fixedname)$.
    \end{enumerate} 
\end{example}

Following Definition~\ref{def:vlambda}, we can see that $\vintsub\viwsub_{i}(x) = \zap{\intsub\iwsub_{i}(x)}$ for any $i \le n$ and $x \in \vars$, but this equality need not lift to bigger terms in general. Consider a minimal $x \in \dom(\intsub)$ with $\intsub(x) = t$. So $t$ is ground, and hence $\varsof(t) = \emptyset$. So $\vintsub\viwsub_{i}(t) = t$. However, $\zap{\intsub\iwsub_{i}(t)} = \zap{t} = \fixedname$, since $t$ is zappable. Thus, it is not true that $\vintsub\viwsub_{i}(t) = \zap{\intsub\iwsub_{i}(t)}$ for all possible terms $t$. However, we can show that this holds for all $t \in \constst$. 

\begin{restatable}{lemma}{zaptaunut}\label{lem:zap-taunut}
    For $i \leq n$ and $t \in \constst$, $\vintsub\viwsub_{i}(t) = \zap{\intsub\iwsub_{i}(t)}$. 
\end{restatable}

We now show, via Lemmas~\ref{lem:dyzetazap} and \ref{lem:eqzetazap}, that small substitutions preserve derivabilities of both terms and equalities. 

\begin{lemma}\label{lem:dyzetazap}
	For $i \leq n$ and any term $t$, 
	if $\intsub(T_{i}) \DYderives t$ then $\vintsub(T_{i}) \DYderives \zap{t}$.
\end{lemma}
\begin{proof}
    Let $X$ and $Y$ stand for $\intsub(T_{i})$ and $\vintsub(T_{i})$. Since $X \subseteq \constst$, by Lemma~\ref{lem:zap-taunut}, $\zap{X} = Y$. Let $\pi$ be a typed normal $\DYderives$ proof of $X \vdash t$ (ensured by Theorem~\ref{thm:rustur}). We prove that $Y \DYderives \zap{t}$. Consider the last rule $\rnrule$ of $\pi$. The following cases arise.
    \begin{itemize}[leftmargin=*]
		\item {$\rnrule = \textsf{ax}$:} $t \in X$, and therefore $\zap{t} \in Y$. Thus $Y \DYderives \zap{t}$ by $\textsf{ax}$. 
		\item {$\rnrule$ is a constructor:} Let $t = \func(t_{1}, t_{2})$ and let the immediate subproofs of $\pi$ be $\pi_{1}, \pi_{2}$, with $\concof(\pi_{i}) = t_{i}$ for $i \le 2$. By IH, there is a proof $\varpi_{i}$ of $Y \vdash \zap{t_{i}}$ for each $i \leq 2$. If $t$ is zappable, then $\zap{t} = \fixedname \in Y$ ($\fixedname \in T_{i}$ for all $i$, so $\fixedname \in X$ and $\fixedname \in Y$), and we have $Y \DYderives \zap{t}$ using $\textsf{ax}$. If $t$ is not zappable, then $\zap{t} = \zap{\func(t_{1}, t_{2})} = \func(\zap{t_{1}}, \zap{t_{2}})$, and we can apply $\rnrule$ on the $\varpi_{i}$s to get $Y \DYderives \zap{t}$.
		\item {$\rnrule$ is a destructor:} Let the immediate subproofs of $\pi$ be $\pi_{1}, \pi_{2}$, deriving $t_{1}, t_{2}$ respectively, with $t_{1}$ being the major premise, and $t$ an immediate subterm of $t_{1}$. Since $\pi$ is typed normal, $\pi_{1}$ is also typed and ends in a destructor, so by Definition~\ref{def:welltyped}, $t_{1} \in \intsub(\stnonvars) \cup \qvar$. Since we applied a destructor on $t_{1}$, it is not in $\qvar$. Thus, there is some $u_{1} \in \stnonvars$, with the same outermost operator as $t_{1}$, such that $t_{1} = \intsub(u_{1})$. Hence, $\bigsub(t_{1}) = \bigsub(u_{1})$. 
		
		If $t_{1}$ were zappable, there would be a minimal $x$ such that $\bigsub(x) = \bigsub(t_{1}) = \bigsub(u_{1}) \in \bigsub(\stnonvars)$, which contradicts the minimality of $x$. Thus, $t_{1}$ is not zappable, and $\zap{t_{1}}$ has the same outermost structure as $t_{1}$. By IH, there is a proof $\varpi_{i}$ of $Y \vdash \zap{t_{i}}$ for each $i \leq 2$. Since $\zap{t_{1}}$ is not atomic, we can apply the destructor $\rnrule$ on the $\varpi_{i}$s to get $Y \DYderives \zap{t}$. 
 \qedhere
	\end{itemize}
\end{proof}

\begin{restatable}{lemma}{eqzetazap}\label{lem:eqzetazap}
    For $i \leq n$ and terms $t,u$, 
    if $\intsub(T_{i};E_{i}) \eqderives \equals{t}{u}$ then $\vintsub(T_{i};E_{i}) \eqderives \equals{\zap{t}}{\zap{u}}$.
\end{restatable}
\begin{proof}
    Let $(X;A)$ and $(Y;B)$ denote $\intsub(T_{i};E_{i})$ and $\vintsub(T_{i};E_{i})$ respectively. As earlier, using Lemma~\ref{lem:zap-taunut}, $\zap{X} = Y$ and $\zap{A} = B$. Let $\pi$ be a typed normal $\eqderives$ proof of $X;A \vdash \equals{t}{u}$ (guaranteed by Theorem~\ref{thm:rustur-eq}). We prove that $Y;B \eqderives \equals{\zap{t}}{\zap{u}}$. Most of the cases are straightforward, so here we only consider the cases when $\pi$ ends in $\rnproj$ or $\rncons$. 
	\begin{itemize}[leftmargin=*]		
	\item {$\pi$ ends in $\rnproj$}: Let the immediate subproof of $\pi$ be $\pi'$ deriving $X;A \vdash \equals{a}{b}$ where $a = \func(a_{0},a_{1})$, $b = \func(b_{0},b_{1})$, and $t = a_{0}$ and $u = b_{0}$. By IH, there is a proof $\varpi'$ of $Y;B \vdash \equals{\zap{a}}{\zap{b}}$. 
        For $\rnproj$, we need $X \DYderives \{a_{0},a_{1},b_{0},b_{1}\}$. By Lemma~\ref{lem:dyzetazap}, $Y \DYderives \{\zap{a_{0}},\zap{a_{1}},\zap{b_{0}},\zap{b_{1}}\}$. By Lemma~\ref{lem:prooftotruth}, 
$\bigsub(a) = \bigsub(b)$. By normality, $\rncons$ cannot occur in $\pi$. $\pi$ is also typed, so either $a = b$ or $a$ and $b$ are typed. If $a = b$, then $t = u$, and we have a proof of $Y;B \vdash \equals{\zap{t}}{\zap{u}}$ ending in $\rneq$. If $a$ and $b$ are typed, we apply Lemma~\ref{lem:projcase} and the following two cases arise.
        \begin{itemize}
            \item {$a$ and $b$ not zappable:} Then $\zap{a}$ and $\zap{b}$ have the same outermost structure as $a$ and $b$, and $\zap{t} = \zap{a_{0}}$ and $\zap{u} = \zap{b_{0}}$. So we can apply $\rnproj$ on $\varpi'$ to get $Y;B \eqderives \equals{\zap{t}}{\zap{u}}$. 
            \item {$a = b$:} Then $t = u$ as well, and hence $\zap{t} = \zap{u}$. Since $Y \DYderives \{\zap{t}, \zap{u}\}$, $Y;B \eqderives \equals{\zap{t}}{\zap{u}}$ with last rule $\rneq$. 
        \end{itemize}
        \item {$\pi$ ends in $\rncons$}: Let $t = \func(t_{0},t_{1})$ and $u = \func(u_{0},u_{1})$. Let $\pi$ have immediate subproofs $\pi_{0}$ and $\pi_{1}$, each $\pi_{i}$ proving $X;A \vdash \equals{t_{i}}{u_{i}}$. By IH, there are proofs $\varpi_{1}, \varpi_{2}$, each $\varpi_{i}$ proving $Y;B \vdash \equals{\zap{t_{i}}}{\zap{u_{i}}}$. By Lemma~\ref{lem:projcase}, two cases arise.
		\begin{itemize}
            \item {$t$ and $u$ not zappable:} Then $\zap{t} = \func(\zap{t_{1}}, \zap{t_{2}})$ and $\zap{u} = \func(\zap{u_{1}}, \zap{u_{2}})$. So $Y;B \eqderives \equals{\zap{t}}{\zap{u}}$ using $\rncons$ on the $\varpi_{i}$s. 
            \item {$t$ and $u$ zappable:} Then, $\zap{t} = \zap{u} = \fixedname \in Y$, so we have a proof of $Y;B \vdash \equals{\zap{t}}{\zap{u}}$ ending in $\rneq$.  \qedhere
        \end{itemize} 
	\end{itemize}
\end{proof}  

Putting Lemmas~\ref{lem:zap-taunut},~\ref{lem:dyzetazap} and~\ref{lem:eqzetazap} together, we get:
\begin{restatable}{theorem}{stot}\label{thm:stot}
    Let $t, u \in \constst$ and $i \le n$. 
    \begin{itemize}[leftmargin=*]
        \item If $\intsub(T_{i-1}) \DYderives \intsub\iwsub_{i}(t)$ then $\vintsub(T_{i-1}) \DYderives \vintsub\viwsub_{i}(t)$.
        \item If $\intsub(T_{i-1}; E_{i-1}) \eqderives \intsub\iwsub_{i}(\equals{t}{u})$ then $\vintsub(T_{i-1}; E_{i-1}) \eqderives \vintsub\viwsub_{i}(\equals{t}{u})$.
    \end{itemize}
\end{restatable}

Having shown that the $\vlambda$s simulate the $\lambda$s, we next show that they allow us a bound on the size of terms therein.
\begin{restatable}{theorem}{vlambdasmall}\label{thm:vlambda-small}
    For $\lambda \in \{\intsub, \bigsub, \iwsub_{i} \mid i \leq n\}$, $\vlambda$ is such that $\dagsize{\vlambda(x)} \le |\stnonvars|$ for all $x \in \dom(\vlambda)$.
\end{restatable}
\begin{proof}
For each $\lambda$ and any $x$, $\vbigsub(\vlambda(x)) = \vbigsub(x) = \zap{\bigsub(x)}$ (by Definition~\ref{def:vlambda}) and thus, $\dagsize{\vlambda(x)} \leq \dagsize{\vbigsub(x)}$. So it suffices to prove a bound on $\dagsize{\vbigsub(x)}$. We show that for $t \in \constst$, $\subterms(\vbigsub(t)) \subseteq \vbigsub(\stnonvars)$. Note that if $t = x$ is non-minimal, there is an $r \in \stnonvars$ s.t.\ $\vbigsub(t) = \vbigsub(r)$. Thus it suffices to prove the statement for $t$ which is either a minimal variable or in $\stnonvars$.

The proof is by induction on $|\vbigsub(t)|$. 
\begin{itemize}[leftmargin=*]
\item $|\vbigsub(t)| = 1:$ $\vbigsub(t) \in \names$. So $t \in \names$ or $t$ is a minimal variable. If $t \in \names$, $\vbigsub(t) = t\in \names$. Otherwise, $\vbigsub(t) = \fixedname$. In both these cases, $\subterms(\vbigsub(t)) \subseteq \vbigsub(\stnonvars)$.
\item $|\vbigsub(t)| > 1:$ Let $a \in \subterms(\vbigsub(t))$. If $a = \vbigsub(u)$ for some $u \in \subterms(t) \setminus \varsof(t)$, then $a \in \vbigsub(\stnonvars)$. If $a = \vbigsub(x)$ for some minimal $x \in \varsof(t)$, then $a = \fixedname = \vbigsub(\fixedname) \in \vbigsub(\stnonvars)$. If $a \in \subterms(\vbigsub(x))$ for non-minimal $x \in \varsof(t)$, then $x \neq t$, and there is an $r \in \stnonvars$ s.t.\ $\vbigsub(x) = \vbigsub(r)$, and $a \in \subterms(\vbigsub(r))$. Since $|\vbigsub(r)| < |\vbigsub(t)|$, by IH, $\subterms(\vbigsub(r)) \subseteq \vbigsub(\stnonvars)$. Thus $a \in \vbigsub(\stnonvars)$. 
\end{itemize}
Hence, $\dagsize{\vbigsub(t)} \leq |\vbigsub(\stnonvars)| \leq |\stnonvars|$, for $t \in \constst$.
\end{proof}

\subsection{NP algorithm for Insecurity: Sketch}
After guessing a coherent set of sessions and an interleaving of these sessions of length $n$, we guess a small substitution $\vintsub$, for each intruder send $\beta_{i}$ a set $X_{i} \subseteq \eatomsof(\beta_{i})$ and a small substitution $\viwsub_{i}$ whose domain is $\boundvars(\beta_{i}) \setminus \boundvars(X_{i})$. We also guess a sequence of knowledge functions such that the relevant atomic assertions and terms (communicated in the $\vintsub(\beta_{i})$s) are derivable from $\vintsub(\dc(\knowfunc_{i-1}(I)))$. These derivability checks in the $\eqderives$ system can be carried out in time polynomial in the size of the protocol description (using the procedure described in Algorithm~\ref{alg:eqderive}). 

For honest agent derivations, we only deal with derivations of the form $\knowfunc_{i}(u_{i}) \assderives \alpha_{i}$ (without the $\intsub$). This is, in fact, a version of the \emph{passive intruder problem} for assertions. Applying Theorem~\ref{thm:eqs-to-gamma}, we reduce this to checks of the form $(U_{i}; F_{i}) \eqderives \hwsub_{i}(\equals{r}{s})$. It is much simpler to ensure that we can obtain $\hwsub_{i}$s of bounded size, because of the absence of $\intsub$. We can think of this as a version of the \emph{passive intruder problem} for the system with assertions. The following theorem, the proof of which can be found in the Appendix, will help us obtain small $\hwsub_{i}$s.   

\begin{theorem}\label{thm:deriv-witness-bounds}
    If there is a $\mu$ satisfying Theorem~\ref{thm:eqs-to-gamma}, there is a ``small'' $\nu$ satisfying the same conditions, such that $\dagsize{\nu(x)} \le |\subterms(S) \cup \subterms(A\cup\{\alpha\})|$ for all $x \in \dom(\nu)$.
\end{theorem}

In order to check whether $\knowfunc_{i}(u_{i}) \assderives \alpha_{i}$, we need to guess $X \subseteq \eatomsof(\alpha_{i})$ and a small substitution $\hwsub_{i}$ such that the conditions of Theorem~\ref{thm:eqs-to-gamma} are satisfied. (The smallness of $\hwsub_{i}$ is guaranteed by Theorem~\ref{thm:deriv-witness-bounds}.) Each of those conditions can be checked in polynomial time because they only involve $\DYderives$ proofs (checkable in PTIME), $\eqderives$ proofs (also checkable in PTIME), and proofs involving only $\{\rnax,\rnconji, \rnexintro, \rnsays\}$ (also checkable in PTIME). Thus, honest agent derivability checks are in NP.

\section{Discussion and Future Work}\label{sec:disc}

\subsection{Intruder theories for terms}
For terms, we assumed that every operator had constructor and destructor rules, as specified in Figure~\ref{fig:consdest}. Such systems are called \emph{constructor-destructor theories}. While the initial results for the active intruder problem were proved for simple theories by~\cite{RT03}, that work has been extended to much richer theories~\cite{CKRT05, AC06, CRZ05, CS03, Bau05, CKR18, CKR20}. As mentioned in Section~\ref{sec:related}, the extension with assertions that we consider is not subsumed by any known intruder theories. 

Can one generalize the results of this paper to richer intruder theories? We believe that one can, but one needs to modify a few fundamental notions used so far. We list these considerations below.

\begin{itemize}
    \item In the main text, we used $\subterm(t)$ to mean the \emph{syntactic subterms} of $t$. For a general intruder theory, we will need to assume a function ${\cal S}$ which maps finite sets of terms to finite sets, and satisfies $\subterm(X) \subseteq {\cal S}(X)$ for any set $X$. 
    \item To handle the general case, we modify the form of constructors and destructors as follows. In a constructor rule, each immediate subterm of the conclusion is a subterm of one of the premises. In a destructor rule, the conclusion is a subterm of one of the premises. 
    \item We can assume that the intruder theory we consider is local w.r.t.~${\cal S}$. That is, whenever $X$ derives $t$, we have a proof $\pi$ of $X \vdash t$ such that $\termsof(\pi) \subseteq {\cal S}(X \cup \{t\})$, and further, if $\pi$ ends in a destructor rule, $\termsof(\pi) \subseteq {\cal S}(X)$. 
    \item We modify Definition~\ref{def:types} to use ${\cal S}$ instead of $\subterms$. Definitions~\ref{def:zappable}, \ref{def:welltyped}, \ref{def:welltypedeq}, \ref{def:zap}, and \ref{def:vlambda}, on which the proofs in Section~\ref{sec:insecurity} hinge, will stay unchanged, since they only refer to $\constst$ and $\stnonvars$.
    \item We need to prove Theorem~\ref{thm:rustur} for the extended theory before moving onto the $\eqderives$ system. Determining the conditions on the intruder theory which would guarantee this theorem is left for future work.
    \item Now, for proofs in the $\eqderives$ system, there is the following subtlety, which we illustrate by considering the $\eqderives$ theory built on top of the theory for {\sc xor} as presented in~\cite{CKRT05}. In this intruder theory, there are implicit rewrites in the rules for {\sc xor}. For instance, from $a \oplus b$ and $b \oplus c$, we can obtain $a \oplus c$. We would need to carry over these rewrites into the equality rules as well, and in the presence of such rewrites, show that normalization and subterm property hold for the new $\eqderives$ system. 

In particular, for normalization, we need to eliminate subproofs where an instance of $\rncons$ appears as the premise for $\rnproj$. For the basic $\eqderives$ system, one can do this by picking the appropriate subproof of $\rncons$. However, in this new system with {\sc xor},  consider a proof of the following form.
\[
    \begin{prooftree}
        \[ 
        	  T; E \vdash \equals{x}{a \oplus b} \quad T; E \vdash \equals{y}{b \oplus c}
        	  \justifies T; E \vdash \equals{x \oplus y}{a \oplus c} \using \rncons
       \]
       \justifies T; E \vdash \equals{x}{a} \using \rnproj_{1}
    \end{prooftree}
\]

Such a proof cannot easily be normalized, since none of these subproofs has the same conclusion. But such a $\rnproj$ rule should not be allowed to begin with, since implicit rewrites are not injective.\footnote{In the constructor-destructor theories as in Figure~\ref{fig:consdest}, we can see that such implicit rewrites do not occur, and all $\func$s considered are injective.} Thus, proving normalization and the subterm property for any modified $\eqderives$ system built on top of a general intruder theory seems feasible, provided one appropriately tailors the rules -- especially $\rnproj$ -- to avoid any unsound behaviour. This is left for future work. 
\end{itemize}

Thus, we can see that the main change in lifting this result to richer intruder theories lies in showing that Theorem~\ref{thm:rustur} holds. One might also need to restrict the new rules one might introduce to the $\eqderives$ system, and hence mildly modify the proofs of the normalization theorem and Theorem~\ref{thm:rustur-eq}.

\subsection{Constraint solving approach}
An algorithmic approach to the active intruder problem is \emph{constraint solving}~\cite{MS01, CS03}. Rather than merely proving a bound on the substitution size, these papers present the problem as a series of deducibility constraints (involving variables), the solution to which is a substitution under which all the deducibilities actually hold. They also provide rules for constructing such a substitution.

In Section~\ref{sec:insecurity}, for a run, we defined the sequence of sets $(T_{i}; E_{i})$, and sets of atomic formulas $X_{i}$, for $i \leq n$. This can be viewed as a generalized constraint system, where we want to find substitutions under which $(T_{i}; E_{i})$ can derive the equality assertions in $X_{i}$, and $T_{i}$ can derive the public terms of $X_{i}$. It is a worthwhile exercise to adapt the existing constraint solving approaches to solve such generalized constraint systems. We leave this for future work.

\subsection{Full disjunction}
An interesting feature of the language in~\cite{RSS17} is the use of disjunction. While our syntax here uses list membership to express a limited form of disjunction that seems to suffice for many protocols, it would be worthwhile to explore the utility of full disjunction and its effect on the active intruder problem. 

In fact, with disjunction, we know that even the derivability problem becomes more involved. To check if $(S;A) \assderives \gamma$, one can no longer work with a single kernel of $(S;A)$. One can define a notion of ``down-closure''. For each disjunctive formula $\alpha\disj\beta$, one obtains two down-closures -- one containing $\alpha$, and the other $\beta$. In general, many disjunctions could occur in $A$, and there are exponentially many down-closures for any $(S; A)$. Using a left disjunction property similar to those in Lemma~\ref{lem:leftprop} ($\alpha\disj\beta$ derives $\gamma$ iff $\gamma$ is derivable from $\alpha$ and from $\beta$), we check if the kernels of all down-closures of $(S;A)$ derive $\gamma$. Thus, the derivability problem is in $\Pi_{2}$. Some of these down-closures might even contain contradictory assertions, and hence our techniques for the insecurity problem do not seem to directly apply. Exploring these issues is an interesting direction of research and is left for future work. 

\subsection{Adding if-then-else branching to protocols}\label{sec:ifthenelse}
As mentioned earlier, we can add an $A:\assertact~\alpha$ action that allows the role to proceed only if $\alpha$ can be derived using the information that $A$ has at the time. Similarly, we can add an action of the form $A:\denyact~\alpha$, which lets the role proceed only if $\alpha$ \emph{cannot be derived} using $A$'s current knowledge. To simulate an if-then-else branch (by specifying a condition $\alpha$ to be checked and an agent $A$ who will check it), we create two roles, one containing $A:\assertact~\alpha$ followed by the actions in the then branch, and the other containing $A:\denyact~\alpha$ followed by the actions in the else branch. We can easily extend our results to protocols involving such assert and deny actions where the condition being checked is whether or not a predicate holds about some atomic terms (for example, $\elg(V)$ in Section~\ref{sec:foo}). 

The fact that a predicate $P$ holds about some terms $\vec{t}$ can be modelled as the presence of $\vec{t}$ in a global list. We can also extend the model to allow agents (with appropriate access privileges) to add and delete entries from global lists, as considered in tools like Proverif~\cite{Bla16} and in some versions of applied-pi~\cite{ALRR17, KK16}. The technical proofs in our work continue to hold for these extensions.

\subsection{Adding assertions to other models and tools}
It is also useful to add communicable assertions to the widely-used applied pi calculus~\cite{ABF17}. It would be especially interesting to see how this impacts the notion of static equivalence, and then study expressibility and decidability. As mentioned earlier, one can express certain ``equivalence'' properties in a more natural manner with assertions as compared to the terms-only model. Another promising extension is to study which equivalence properties can be expressed as reachability properties in this manner, like in~\cite{GMV22}. These would also help us to extend existing tools~\cite{Cre08, MSCB13, Bla16, CKR18} with assertions.

\bibliographystyle{plain}
\bibliography{ref}

\newpage
\appendix 

\section{Proof of Theorem~\ref{thm:deriv-witness-bounds}}\label{app:decidable-assderives}
We want to check if $(S;A) \assderives \alpha$, where $\boundvars(\alpha) \cap \varsof(S;A) = \emptyset$. Let $(T;E) = \dc(S;A)$. By Theorem~\ref{thm:eqs-to-gamma}, this reduces to checking if there is a substitution $\mu$ with $\dom(\mu) = \boundvars(\alpha)$ s.t. and $X \subseteq \eatomsof(\alpha)$ s.t. 
$\forall{}x \in \boundvars(\alpha): T \DYderives \mu(x)$ and 
for all $\beta \in X$, $(T;E) \assderives \mu(\beta)$. For formulas in $X$ that are not of the form $\equals{t}{u}$, all terms occurring in them are variables or names, so $\mu$ is atomic on variables occurring in them. It therefore suffices to only consider assertions of the form $\equals{t}{u}$. 

So the problem is as follows. There is a set of terms $\constst$ and $(T; E)$ with $\subterms(T) \cup \subterms(E) \subseteq \constst$, and a substitution $\mu$ with $\dom(\mu) \cap \varsof(T;E) = \emptyset$, which satisfies some derivabilities of the form $T \DYderives t$ and $T;E \eqderives \equals{t}{u}$, where $t,u \in \constst$. We seek a small $\nu$ that preserves the above derivabilities. 
To reduce clutter, we use $\dommu$ to refer to $\dom(\mu)$. Let $\stnonvars = \constst \setminus \dommu$. Since $T \DYderives \mu(x)$, all variables occurring in $\mu(x)$ must also be in $\varsof(T)$. But $\varsof(T;E) \cap \dommu = \emptyset$, so $\varsof(\mu(x)) \cap \dommu = \emptyset$. 

Define $t \approx u$ iff $T;E \eqderives \mu(\equals{t}{u})$. It is easy to see that $\approx$ is a partial equivalence relation (on the subset of terms $t$ such that $T \DYderives \mu(t)$).   

We say that $x \in \dommu$ is \emph{minimal} if there is no $t \in \stnonvars$ with $x \approx t$. Let $\vars_{m}$ denote the set of all minimal variables. Our strategy for finding a small $\nu$ is to ``zap'' minimal variables, and propagate the change to (interpretations of) non-minimal variables. To this end, it is convenient to translate every term to an ``equivalent'' one with only minimal variables. The notion of equivalence is based on unifiability under $\mu$. The set of all such terms that are equivalent to terms in $\constst$ is defined as follows. 

\begin{definition}
    $\hatst \coloneqq \{t \mid \varsof(t) \cap \dommu \subseteq \vars_{m}, \text{ either } t \in \vars_{m} \text{ or } \exists{u} \in \stnonvars: t \approx u\}$. 
\end{definition}

\begin{lemma}\label{lem:st-to-hatst} 
    For every $t \in \constst$ with $T \DYderives \mu(t)$, there is $\expplus{t} \in \hatst$ such that: $T \DYderives \mu(\expplus{t})$; $t \approx \expplus{t}$; and for all $x \in \vars_{m}$, $\posof{x}{\expplus{t}} \subseteq \abstractable(T \cup \dommu, \expplus{t})$. 
\end{lemma}
\begin{proof}
    For $x, y \in \dommu$, $x \prec y$ iff $\exists{r}\in \stnonvars[x \in \subterms(r) \text{ and } r \approx y]$. 

    We now show that $\prec$ is acyclic. Towards this, we claim that if $x \prec y$ and $y \prec z$, then there is some term $a$ (not necessarily in $\constst$) s.t.\ $\mu(x)$ is a proper subterm of $\mu(a)$ and $a \approx z$. Extending this reasoning, we see that if $x \prec^{+} x$, we have some term $a$ such that $\mu(x)$ is a proper subterm of $\mu(a)$ and $(T;E) \eqderives \equals{\mu(a)}{\mu(x)}$. But $E$ is consistent, which means that there is some $\lambda$ s.t.\ $\lambda(\mu(a)) = \lambda(\mu(x))$. But this is incompatible with $\mu(x)$ being a proper subterm of $\mu(a)$. Thus $\prec$ is acyclic. 

    We now prove the claim. Suppose $x \prec y$ and $y \prec z$. Then there exists $r,s \in \stnonvars$ such that $x \in \subterms(r)$, $(T;E) \eqderives \equals{\mu(r)}{\mu(y)}$, $y \in \subterms(s)$, and $(T;E) \eqderives \equals{\mu(s)}{\mu(z)}$. Let $a = \replsubtermat{s}{\posof{y}{s}}{r}$. We see that $\mu(x)$ is a proper subterm of $\mu(a)$. From the abstractability conditions satisfied by $\mu$ and the derivability of $\mu(x)$ for all $x \in \dommu$, we can justify the applications of $\rnsubst$ necessary to obtain $(T;E) \eqderives \equals{\mu(a)}{\mu(z)}$ and thus $a \approx z$. 

    Since $\prec$ is acyclic, we can define a notion of \emph{rank} of variables: $\rank(x) = \max\{\rank(y) \mid y \prec^{+} x\} + 1$. For a term $u \in \stnonvars$, we define $\rank(u) = \max\{\rank(x) \mid x \in \varsof(u) \cap \dommu\}$. It is easy to verify that if $u \in \stnonvars$ and $x \approx u$, then $\rank(x) > \rank(u)$. It is also easy to see that if $x \in \vars_{m}$, then $x \in \dommu$ has rank $0$.  

    Having set up this machinery, we prove the lemma by induction on $\delta(t) = (\rank(t), |t|)$. First fix an ordering on $\hatst$. For $\delta(t) = (0,0)$, we have that $t$ is a variable $x$ and $\rank(x) = 0$. We have two cases to consider.
    \begin{itemize} 
        \item {\bf $x \in \vars_{m}$:} Choose $\expplus{x} = x$.
        \item {\bf $x \notin \vars_{m}$:} This means that there is some $u \in \stnonvars$ s.t.\ $x \approx u$. But since $\rank(x) = 0$, $\varsof(u) \cap \dommu = \emptyset$ for each such $u$. Choose $\expplus{x}$ to be the earliest such $u$ (according to the ordering on $\hatst$). Clearly $(T;E) \vdash \equals{\mu(x)}{\mu(\expplus{x})}$, and by Lemma~\ref{lem:dcpure}, $T \DYderives \mu(\expplus{x})$. Finally $\varsof(\expplus{x}) \cap \dommu = \emptyset$, so it is vacuously true that $\posof{y}{\expplus{x}} \subseteq \abstractable(T \cup \dommu, \expplus{x})$ for all $y \in \vars_{m}$. 
    \end{itemize}
    So suppose $\delta(t) > (0,0)$ and that the theorem is true for all $u$ such that $\delta(u) < \delta(t)$. There are two cases to consider:
    \begin{itemize}
        \item {$t$ is a variable, say $x$:} Then $\rank(x) > 0$, and there is $u \in \stnonvars$ s.t.\ $x \approx u$, whence $\rank(u) < \rank(x)$. Pick the earliest such $u \in \hatst$. By IH there is $\expplus{u}$, and we define $\expplus{x} = \expplus{u}$. Since $x \approx u$ and $u \approx \expplus{u}$, we have $x \approx \expplus{x}$, by transitivity.
        \item {$t$ is not a variable:} For each $y \in \varsof(t) \cap \dommu$, there is $\expplus{y}$. We obtain $\expplus{t}$ by replacing each $y$ by $\expplus{y}$. Clearly $\varsof(\expplus{t}) \cap \dommu \subseteq \vars_{m}$. Also since all variables appear in abstractable positions of $t$, we can justify the relevant applications of $\rnsubst$ to justify $t \approx \expplus{t}$. Finally, if $z$ appears in an abstractable position in $r$ and $y$ appears in an abstractable position in $s$, then $z$ appears in an abstractable position in $\replsubtermat{s}{\posof{y}{s}}{r}$. Thus the abstractability part of the statement is also fulfilled.  \qedhere
    \end{itemize}
\end{proof} 

We now define the substitution $\nu$ as follows. Assume that there is some $\fixedname \in T \cap \names$ such that $\fixedname \notin \subterms(E \cup \{\alpha\}) \cup \subterms(\rng(\mu))$.\footnote{Thus $\fixedname$ is a ``spare name'' that does not occur in any of the derivations under consideration.} Let $\nu_{m}$ be the substitution that maps each $x \in \vars_{m}$ to $\fixedname$. For all $x \in \dommu: \nu(x) = \nu_{m}(\expplus{x})$. Notice that for all $x \in \dom(\nu)$, either $\nu(x) = \fixedname$ or there is $u \in \stnonvars$ s.t.\ $\nu(x) = \nu(u)$. Thus we can show that $\nu$ is $|\constst|$-bounded following the proof of Theorem~\ref{thm:vlambda-small}. To complete the proof of Theorem~\ref{thm:deriv-witness-bounds}, we just need to show that $\nu$ preserves derivability. This is proved in Theorem~\ref{thm:nu-simulate-mu}, the main result of this section. But first we state a useful observation. 
\begin{obs}\label{obs:constst-not-mvars}
    \phantom{a}
    \begin{enumerate}
        \item For $x \in \dommu$, if $\mu(x) \in \constst$ then $x \notin \vars_{m}$. 
        \item If $t \in \hatst$ and $\mu(t) \in \constst$, then $\varsof(t) \cap \dommu = \emptyset$ and $\mu(t) = t$.  
    \end{enumerate}
\end{obs}
\begin{proof}
    \phantom{a} 
    \begin{enumerate}
        \item Let $\mu(x) = t \in \constst$. Since $\varsof(t) \cap \dommu = \emptyset$, we have that $t \notin \dommu$ and $\mu(t) = t$. Thus $t \in \stnonvars$, and $\equals{\mu(x)}{\mu(t)}$ is derivable using the $\rneq$ rule, i.e., $x \approx t$. Therefore $x \notin \vars_{m}$.
        \item For every $x \in \varsof(t) \cap \dommu$, $\mu(x) \in \constst$. Thus we have $x \notin \vars_{m}$, by the previous part. But since $t \in \hatst$, we have that $\varsof(t) \cap \dommu \subseteq \vars_{m}$. The only conclusion is that $\varsof(t) \cap \dommu = \emptyset$, and thus $\mu(t) = t$. 
        \qedhere
    \end{enumerate}
\end{proof}

\begin{theorem}\label{thm:nu-simulate-mu} 
    \phantom{a}
    \begin{enumerate}
        \item For any $t \in \constst$, if $T \DYderives \mu(t)$ then $T \DYderives \nu(t)$. 
        \item For any $t, u \in \constst$, if $T; E \eqderives \equals{\mu(t)}{\mu(u)}$ then $T; E \eqderives \equals{\nu(t)}{\nu(u)}$. 
    \end{enumerate}
\end{theorem}
\begin{proof}
    By Lemma~\ref{lem:st-to-hatst}, it suffices to prove the following. Let $r,s \in \hatst$ such that $\forall{x} \in \vars_{m}$, $\posof{x}{(r,s)} \subseteq \abstractable(T \cup \dommu, (r,s))$. If $T \DYderives \mu(r)$ then $T \DYderives \nu_{m}(r)$; and if $T;E \eqderives \equals{\mu(r)}{\mu(s)}$ then $T;E \derives \equals{\nu_{m}(r)}{\nu_{m}(s)}$. 

    \begin{enumerate}
        \item Suppose $T \DYderives r$ for $r$ as above. Since all positions of variables from $\dommu$ occurring in $r$ are abstractable w.r.t. $T \cup \dommu$, and since $T \cup \{\fixedname\} \DYderives \fixedname$, we can easily prove by induction on the size of terms that $T \cup \fixedname \DYderives \nu_{m}(r)$. 
        \item Suppose $T;E \eqderives \equals{\mu(r)}{\mu(s)}$ for $r,s$ as above. Let $\pi$ be a normal proof of $T;E \vdash \equals{\mu(r)}{\mu(s)}$ with last rule $\rnrule$. We prove the desired statement by induction on the structure of $\pi$. There are the following cases to consider. 
        \begin{itemize}
            \item {$\rnrule \in \{\rnax, \rneq, \rnproj\}$:} Three cases arise: $\equals{\mu(r)}{\mu(s)} \in E$, and thus $\mu(r), \mu(s) \in \constst$. Or $\mu(r) = \mu(s)$ and $T \DYderives \mu(r)$ via a proof ending in ${\sf ax}$ or a destructor rule, and thus $\mu(r), \mu(s) \in \subterms(T) \subseteq \constst$. Or by subterm property for normal $\eqderives$-proofs $\mu(r), \mu(s) \in \subterms(T;E) \subseteq \constst$. Thus $\mu(r), \mu(s) \in \constst$ in all three cases. By Observation~\ref{obs:constst-not-mvars}, $\varsof(r,s) \cap \dommu = \emptyset$. Thus $\nu_{m}(r) = r = \mu(r)$ and $\nu_{m}(s) = s = \mu(s)$. Therefore $\pi$ itself is a proof of $\equals{\nu_{m}(r)}{\nu_{m}(s)}$. 
            
            \item {$\rnrule = \rntrans$:} Suppose the immediate subproofs are $\pi_{1}, \ldots, \pi_{n}$, with each $\pi_{i}$ deriving $\equals{v_{i-1}}{v_{i}}$. Let $\mu(r) = v_{0}$ and $\mu(s) = v_{n}$. Since no $\pi_{i}$ ends in $\rntrans$ and no two adjacent $\pi_{i}$'s end in $\rncons$, each $v_{i}$ (for $0 < i < n$) appears in at least one proof ending in $\rnax$, $\rneq$ or $\rnproj$. Thus, by the subterm property, $v_{i} \in \subterms(T;E) \subseteq \constst$ for $0 < i < n$. Since $\varsof(T;E) \cap \dommu = \emptyset$, it follows that $v_{i} \in \hatst$ and $\mu(v_{i}) = v_{i}$. Thus we can view each $\pi_{i}$ as deriving $\equals{\mu(r_{i-1})}{\mu(r_{i})}$, where $r_{i-1}, r_{i} \in \hatst$ (taking $r_{0}$ and $r_{n}$ to be $r$ and $s$). By IH, there are proofs $\varpi_{1}, \ldots, \varpi_{n}$, with each $\varpi_{i}$ deriving $\equals{\nu_{m}(r_{i-1})}{\nu_{m}(r_{i})}$. By composing them using $\rntrans$, we get a proof of $T;E \vdash \equals{\nu_{m}(r)}{\nu_{m}(s)}$, as desired.

            \item {$\rnrule = \rncons$:} Suppose $r = \func(r_{1}, \ldots, r_{n})$ and $s = \func(s_{1}, \ldots, s_{n})$. Each $r_{i}, s_{i} \in \hatst$, and the immediate subproofs are $\pi_{1}, \ldots, \pi_{n}$, with each $\pi_{i}$ deriving $\equals{\mu(r_{i})}{\mu(s_{i})}$. By IH we have proofs $\varpi_{1}, \ldots, \varpi_{n}$, with each $\varpi_{i}$ proving $\equals{\nu_{m}(r_{1})}{\nu_{m}(s_{1})}$. We can compose them with the $\rncons$ rule to get the desired proof of $\equals{\nu_{m}(r)}{\nu_{m}(s)}$. 

            Suppose, on the other hand, that $r$ is a variable. Since $r \in \hatst$, $r \in \vars_{m}$. Now $s \in \hatst$, so either $s \in \vars_{m}$ or there is $a \in \stnonvars$ with $s \approx a$. But in the second case, $r \approx a$ (by symmetry and transitivity), which cannot happen for a minimal variable $r$. Therefore $s \in \vars_{m}$. And we have $\nu_{m}(r) = \nu_{m}(s) = \fixedname \in T$, so there is a proof of $T, E \eqderives \equals{\nu_{m}(r)}{\nu_{m}(s)}$ ending in $\rneq$. 

            We have a similar argument in case $s$ is a variable, thereby proving the theorem. \qedhere
        \end{itemize}
    \end{enumerate}
\end{proof}

\section{Normalization and subterm property for \texorpdfstring{$\eqderives$}{derivation of equalities}}
\label{app:normalization}

Suppose $E \cup \{\alpha\}$ consist only of atomic formulas and $\pi$ is a proof of $T; E \eqderives \alpha$. We say that $\pi$ is \emph{normal} if the following hold.
	\begin{enumerate}
		\item All $\DYderives$ subproofs are normal. 
		\item The premise of $\rnsymm$ can only be the conclusion of $\rnax$ or $\rnprom$.
		\item The premise of $\rneq$ can only be the conclusion of a destructor rule.
		\item No premise of a $\rntrans$ is of the form $\equals{a}{a}$, or the conclusion of a $\rntrans$.
		\item Adjacent premises of a $\rntrans$ are not conclusions of $\rncons$.
		\item No premise of $\rnlint$ is the conclusion of $\rnlint$ or $\rnlwk$.
		\item No subproof ending in $\rnproj$ contains $\rncons$. 
	\end{enumerate}

A set $E$ of atomic formulas is said to be \emph{consistent} if there is a $\lambda$ s.t.\ $\lambda(t) = \lambda(u)$ for each $\equals{t}{u} \in E$, and $\lambda(t) \in \{t_{1},\ldots,t_{n}\}$ for each $t\listmemb{[t_{1},\ldots,t_{n}]} \in E$. 

\begin{table}[!t]
    \begin{center}
        \bgroup
        \def\arraystretch{1.25}
        \begin{tabular}{|l|c|}
            \hline
            \multirow{2}{*}{R1} 
            & $\rneq(\func(\pi_{1},\pi_{2}))$ \\  
            & $\rncons_{\func}(\rneq(\pi_{1}), \rneq(\pi_{2}))$ \\
            \hline 
            \multirow{2}{*}{R2} 
            & $\rnsymm(\rneq(\pi))$ \\ 
            & $\rneq(\pi)$ \\
            \hline 
            \multirow{2}{*}{R3} 
            & $\rnsymm(\rnsymm(\pi))$ \\
            & $\pi$ \\ 
            \hline 
            \multirow{2}{*}{R4} 
            & $\rnsymm(\rnrule(\pi_{1}, \ldots, \pi_{k}))$ \\ 
            & $\rnrule(\rnsymm(\pi_{1}), \ldots, \rnsymm(\pi_{k}))$ \\ 
            \hline 
            \multirow{2}{*}{R5}
            & $\rntrans(\pi_{1}, \ldots, \pi_{i-1}, \varpi, \pi_{i}, \ldots, \pi_{r-1})$ \\
            & $\rntrans(\pi_{1}, \ldots, \pi_{i-1}, \pi_{i}, \ldots, \pi_{r-1})$ \\ 
            \hline 
            \multirow{2}{*}{R6} 
            & $\rntrans(\pi_{1}, \ldots, \rntrans(\pi^{1}_{i}, \ldots, \pi^{k}_{i}), \ldots, \pi_{r-1})$ \\
            & $\rntrans(\pi_{1}, \ldots, \pi^{1}_{i}, \ldots, \pi^{k}_{i}, \ldots, \pi_{r-1})$ \\
            \hline 
            \multirow{2}{*}{R7} 
            & $\rntrans(\pi_{1}, \ldots, \rncons(\pi^{1}_{i-1}, \pi^{2}_{i-1}), \rncons(\pi^{1}_{i}, \pi^{2}_{i}), \ldots, \pi_{r-1})$ \\
            & $\rntrans(\pi_{1}, \ldots, \rncons(\rntrans(\pi^{1}_{i-1}, \pi^{1}_{i}), \rntrans(\pi^{2}_{i-1}, \pi^{2}_{i})), \ldots, \pi_{r-1})$ \\
            \hline 
            \multirow{2}{*}{R8} 
            & $\rnproj_{j}(\rncons(\pi_{1},\pi_{2}))$ \\
            & $\pi_{j}$ \\
            \hline 
            \multirow{2}{*}{R9} 
            & $\rnproj_{j}(\rntrans(\pi_{1}, \ldots,\pi_{i-1}, \rncons_{\func}(\pi^{1}_{i},\pi^{2}_{i}), \pi_{i+1}, \ldots, \pi_{r-1}))$ \\ 
            & $\rntrans(\rnproj_{j}(\rntrans(\pi_{1}, \ldots,\pi_{i-1})), \pi^{j}_{i}, \rnproj_{j}(\rntrans(\pi_{i+1}, \ldots, \pi_{r-1})))$ \\
            \hline 
            \multirow{2}{*}{R10} 
            & $\rnlint(\pi_{1}, \ldots, \pi_{k-1}, \rnlint(\pi_{k}, \ldots, \pi_{m}), \pi_{m+1}, \ldots, \pi_{n})$ \\
            & $\rnlint(\pi_{1}, \ldots, \pi_{k-1}, \pi_{k}, \ldots, \pi_{m}, \pi_{m+1}, \ldots, \pi_{n})$ \\ 
            \hline 
            \multirow{2}{*}{R11} 
            & $\rnlint(\pi_{1}, \ldots, \rnweak(\pi_{i}), \ldots, \pi_{n})$ \\ 
            & $\rnweak(\pi_{i})$ \\
            \hline
        \end{tabular}
        \egroup
    \end{center}
    \caption{Proof transformation rules. The proof represented by the first line in each row is transformed to the proof represented by the second line. In R4, $\rnrule \in \{\rntrans, \rnproj, \rncons\}$. In R5, $\concof(\varpi)$ is assumed to be of the form $\equals{a}{a}$.}        
    \label{tab:rwrules}
\end{table}

   We next prove normalization for $\eqderives$ proofs (with a consistent LHS). We present proof transformation rules in Table~\ref{tab:rwrules}. To save space, we use \emph{proof terms} -- $\rnrule (\pi_{1}, \ldots, \pi_{n})$ denotes a proof $\pi$ with last rule $\rnrule$ and immediate subproofs $\pi_{1}, \ldots, \pi_{n}$. It is assumed that the derivations are from a consistent $(T;E)$. R1 is applicable when $\func$ is a constructor rule, and ensures that $\DYderives$ subproofs do not end in a constructor rule. R2 and R3 eliminate some occurrences of $\rnsymm$, while R4 pushes $\rnsymm$ up towards the axioms. R5 and R6 ensure that no premise of $\rntrans$ is the conclusion of $\rneq$ or $\rntrans$. R7 ensures that adjacent premises of $\rntrans$ are not the result of $\rncons$. R8 simplifies proofs where $\rnproj$ follows $\rncons$. We will discuss R9 later. R10 ensures that the conclusion of $\rnlint$ is not a premise of $\rnlint$. In R11, $\pi_{i}$ proves an equality $\equals{v}{n}$, and it is weakened to a list membership of the form $v \listmemb \ell'$, but by consistency, even after intersection, the conclusion must be of the form $v\listmemb{\ell}$ where $\lambda(v)$ is an element of $\ell$ for some $\lambda$. Thus we can directly apply weakening to $\pi_{i}$ to get the same conclusion. 
        
    R9 requires some explanation. Let $\pi_{i}$ be the proof $\rncons_{\func}(\pi^{1}_{i}, \pi^{2}_{i})$, and let $\concof(\pi_{j})$ be $\equals{t_{j}}{t_{j+1}}$, for $1 \leq j < r$. We see that $\concof(\rntrans(\pi_{1}, \ldots, \pi_{r-1}))$ is $\equals{t_{1}}{t_{r}}$. Since $\rnproj$ is applied on this, there is some constructor $\gunc$ such that $t_{e} = \gunc(t^{1}_{e},  t^{2}_{e})$ for $e \in \{1, r\}$. Since $\pi_{i}$ ends in $\rncons_{\func}$, we see that $t_{e} = \func(t^{1}_{e}, t^{2}_{e})$ for $e \in \{i, i+1\}$. But $\equals{t_{1}}{t_{i}}$ is provable from $(T;E)$, which is consistent. Therefore it has to be the case that $\func = \gunc$. Thus we see that for all $e \in \{1, i, i+1, r\}$, $t_{e} = \func(t^{1}_{e}, t^{2}_{e})$. So we can rewrite the LHS of R9 to the RHS to get a valid proof. Note that we can apply $\rnproj$ on $\equals{t_{1}}{t_{i}}$ in the transformed proof since all components of $t_{1}$ and $t_{i}$ are abstractable -- for $t_{1}$ this is true because the $\rnproj$ rule was applied to $\equals{t_{1}}{t_{r}}$ in the proof on the LHS; and for $t_{i}$ this follows from the fact that $\pi^{1}_{i}$ (resp.\ $\pi^{2}_{i}$) derives $\equals{t^{1}_{i}}{t^{1}_{i+1}}$ (resp.\ $\equals{t^{2}_{i}}{t^{2}_{i+1}}$), and so by Lemma~\ref{lem:dcpure}, $T \DYderives \{t^{1}_{i}, t^{2}_{i}\}$. For a similar reason, we can apply $\rnproj$ on $\equals{t_{i+1}}{t_{r}}$.

\begin{theorem} If $(T;E) \eqderives \alpha$ then there is a normal proof of $(T;E) \vdash \alpha$ in the $\eqderives$ system.
\end{theorem}
\begin{proof}
Let $\pi$ be any proof of $(T; E) \vdash \alpha$ such that all DY subproofs of $\pi$ are normal. 
    Suppose we repeatedly apply the transformations of Table~\ref{tab:rwrules} starting with $\pi$ and reach a proof $\varpi$ on which we can no longer apply any of the rules. Then $\varpi$ satisfies clauses 1 to 6 in the definition of normal proofs (since none of the rewrite rules, in particular R1--R7 and R10--R11, apply to $\varpi$). 
            
    Clause 7 is also satisfied by $\varpi$, for the following reason. Suppose a subproof $\varpi_{1}$ ends in $\rnproj$ and $\varpi_{2}$ is a maximal subproof of $\varpi_{1}$ ending in $\rncons$. $\varpi_{2}$ is a proper subproof of $\varpi_{1}$, so there has to be a subproof of $\varpi_{1}$ of the form $\rho = \rnrule(\cdots\varpi_{2}\cdots)$. Since $\rncons$ appears as the rule above $\rnrule$, a priori, $\rnrule$ can only be one of $\{\rnsymm, \rntrans, \rnproj, \rncons\}$. But since $\varpi_{2}$ is a \emph{maximal subproof} of $\varpi_{1}$ ending in $\rncons$, $\rnrule \neq \rncons$. Since R4 and R8 cannot be applied on $\varpi$, $\rnrule \notin \{\rnsymm, \rnproj\}$. But if $\rnrule = \rntrans$, then $\rho$ is a proper subproof of $\varpi_{1}$. In particular, it is the immediate subproof of some $\rho' = \rnrule'(\cdots\rho\cdots)$. Now $\rnrule'$ cannot be $\rnsubst$, since then $\concof(\rho')$ is a list membership assertion, which cannot occur in a proof ending in $\rnproj$. $\rnrule' \neq \rncons$, as that would violate the maximality of $\varpi_{2}$. $\rnrule' \notin \{\rnsymm, \rntrans, \rnproj\}$, since then one of the rewrite rules R4, R6, R8 would apply to $\varpi$. We have ruled out all possible cases for $\rnrule'$, and thus we are forced to conclude that $\varpi_{2}$ cannot be a subproof of $\varpi_{1}$. Thus, $\rncons$ does not occur in any subproof of $\varpi$ ending in $\rnproj$, and $\varpi$ satisfies all the clauses in the definition of normal proofs. 
    
    We next show that we can always reach a stage where no transformation is enabled. To begin with, apply the rules R2--R4 until the premise of each occurrence of $\rnsymm$ is the conclusion of an $\rnax$ or a $\rnprom$. None of the other rules converts a proof ending in $\rnax$ or $\rnprom$ to one which does not, so the above property is preserved even if we apply the other rules in any order. 
    
    Associate three sizes to an $\eqderives$-proof $\pi$: 
    \begin{itemize}
    \item $\measure_{1}(\pi)$ is the sum of the sizes of the $\DYderives$ subproofs of $\pi$, 
    \item $\measure_{2}(\pi)$ is the number of $\rncons$ rules that occur in $\pi$, and 
    \item $\measure_{3}(\pi)$ is the size of the proof $\pi$ (number of nodes in the proof tree). 
    \end{itemize}
    
    We also define $\measure(\pi) \coloneqq (\measure_{1}(\pi), \measure_{2}(\pi), \measure_{3}(\pi))$.

    We now show that if $\pi'$ is obtained from $\pi$ by one application of any of the transformation rules other than R2--R4, $\measure(\pi') < \measure(\pi)$.
    \begin{itemize}
        \item If R1 is applied, $\measure_{1}(\pi') < \measure_{1}(\pi)$ and so $\measure(\pi') < \measure(\pi)$.
        \item If R7 or R9 is applied, we have $\measure_{1}(\pi') \leq \measure_{1}(\pi)$ and $\measure_{2}(\pi') < \measure_{2}(\pi)$. Therefore, $\measure(\pi') < \measure(\pi)$.
        \item If R5, R6, R8, R10 or R11 is applied, we have that $\measure_{i}(\pi') \leq \measure_{i}(\pi)$ for $i \in \{1,2\}$ and $\measure_{3}(\pi') < \measure_{3}(\pi)$. So $\measure(\pi') < \measure(\pi)$.
    \end{itemize}
    Thus, once we apply R2--R4 till they can no longer be applied, we cannot have an infinite sequence of transformations starting from any $\pi$. Hence, every proof $\pi$ can be transformed into a normal proof $\varpi$ with the same conclusion. 
\end{proof}

We state and prove the subterm property next. 

\begin{theorem}
For any normal proof $\pi$ of $T; E \eqderives \alpha$, \\
$\termsof(\pi) \subseteq \subterms(T) \cup \subterms(E\cup\{\alpha\})$, and \\
$\listsof(\pi) \subseteq \listsof(E \cup \{\alpha\}) \cup \{[n] \mid n \in \subterms(T) \cup \subterms(E\cup\{\alpha\})\}$.
If $\pi$ does not contain $\rncons$, then $\termsof(\pi) \subseteq \subterms(T) \cup \subterms(E)$ . Also, if $\pi$ does not end in $\rnweak$ and does not end in $\rnlint$, then $\listsof(\pi) \subseteq \listsof(E) \cup \{[n] \mid n \in \subterms(T) \cup \subterms(E)\}$.
\end{theorem}
We implicitly use the following easily provable facts. 
\begin{enumerate}[label=(F\arabic*)]
\item \label{item:f1} If a normal proof $\pi$ ends in $\rntrans$ and an immediate subproof $\varpi$ does not end in $\rncons$, then $\rncons$ does not occur in $\varpi$. 
\item \label{item:f2} If a normal proof $\pi$ derives a list membership assertion, $\rncons$ does not occur in $\pi$.
\end{enumerate}

\begin{proof}
    Let $\rnrule$ be the last rule of $\pi$. We have the following cases. We mention $\listsof(\pi)$ only in cases where the rules involve lists. 
    \begin{itemize}
        \item $\rnrule = \rnax$: $\alpha \in E$, so $\termsof(\pi) \subseteq \subterms(E)$ and $\listsof(\pi) \subseteq \listsof(E)$. 
        \item $\rnrule = \rneq$: $\alpha$ is $\equals{t}{t}$ and $T \DYderives t$. Since $\pi$ is a normal proof whose $\DYderives$ subproofs are also normal, $T \DYderives t$ via a proof ending in a destructor rule, and by subterm property for $\DYderives$, it follows that $t \in \subterms(T)$. Thus $\termsof(\pi) = \{t\} \subseteq \subterms(T)$.
        \item $\rnrule = \rnsymm$: $\termsof(\pi) = \termsof(\pi')$, where $\pi'$ is the immediate subproof, and the statement follows by IH.
        \item $\rnrule = \rncons$: $\alpha$ is $\equals{\func(t_{1},t_{2})}{\func(u_{1},u_{2})}$, and for $i \in \{1,2\}$, there is a subproof $\pi_{i}$ with conclusion $\equals{t_{i}}{u_{i}}$. By IH, $\termsof(\pi_{i}) \subseteq \subterms(T \cup \{t_{i},u_{i}\}) \cup \subterms(E) \subseteq \subterms(T) \cup \subterms(E \cup \{\alpha\})$ for $i \in \{1,2\}$. Thus $\termsof(\pi) \subseteq \subterms(T) \cup \subterms(E \cup \{\alpha\})$.
        \item $\rnrule = \rntrans$: Suppose the subproofs of $\pi$ are $\pi_{1}$ through $\pi_{k-1}$ with conclusions $\equals{t_{1}}{t_{2}}$ through $\equals{t_{k-1}}{t_{k}}$ respectively, and $\alpha = \equals{t_{1}}{t_{k}}$. Since $\pi$ is a normal proof, no two adjacent premises of $\rnrule$ are obtained by $\rncons$, and no premise of $\rnrule$ is obtained by $\rntrans$. 
        The following cases arise.
        \begin{itemize}
            \item $r \in \{t_{1},t_{k}\}$. In this case, $r \in \subterms(\alpha)$.
            
            \item $r \in \termsof(\pi_{i})$, where $\pi_{i}$ does not end in $\rncons$. By~\ref{item:f1}, $\rncons$ does not occur in $\pi_{i}$. By IH, $r \in \subterms(T) \cup \subterms(E)$. 
            
            \item $r \in \termsof(\pi_{i})$, where $\pi_{i}$ ends in $\rncons$, and $1 < i < k-1$. Both $\pi_{i-1}$ and $\pi_{i+1}$ end in a rule other than $\rncons$, by normality of $\pi$. So, by~\ref{item:f1}, $\rncons$ does not occur in $\pi_{i-1}$ and $\pi_{i+1}$, and $t_{i}, t_{i+1} \in \termsof(\pi_{i-1}) \cup \termsof(\pi_{i+1}) \subseteq \subterms(T) \cup \subterms(E)$ (by IH on $\pi_{i-1}$ and $\pi_{i+1}$). So, by applying IH on $\pi_{i}$, we get $r \in \subterms(T) \cup \subterms(E \cup \{\equals{t_{i}}{t_{i+1}}\}) \subseteq \subterms(T) \cup \subterms(E)$.
            
            \item $r \in \termsof(\pi_{1})$, where $\pi_{1}$ ends in $\rncons$. By normality of $\pi$, we see that $\pi_{2}$ ends in a rule other than $\rncons$. So $\rncons$ does not occur in $\pi_{2}$. By IH on $\pi_{2 }$, $t_{2} \in \termsof(\pi_{2}) \subseteq \subterms(T) \cup \subterms(E)$. By IH on $\pi_{1}$, $r \in \subterms(T\cup\{t_{1}, t_{2}\}) \cup \subterms(E) \subseteq \subterms(T) \cup \subterms(E \cup \{\alpha\})$.

            \item $r \in \termsof(\pi_{k-1})$, where $\pi_{k-1}$ ends in $\rncons$. The proof is similar to the above. 
        \end{itemize}        
        \item $\rnrule = \rnproj$: Let $\alpha = \equals{t}{u}$, got from a proof $\pi'$ with conclusion $\equals{a}{b}$. Since $\pi$ is normal, $\rncons$ does not occur in $\pi$ (or in $\pi'$). By IH, $a,b \in \termsof(\pi') \subseteq \subterms(T) \cup \subterms(E)$. Since $t, u \in \subterms(\{a,b\})$, we have $\termsof(\pi) \subseteq \subterms(T) \cup \subterms(E)$.
        \item $\rnrule = \rnprom$: $\alpha$ is $\equals{t}{u}$, and the immediate subproof $\pi'$ proves $t\listmemb{[u]}$. $\pi'$ does not contain $\rncons$, and so by IH, $\termsof(\pi) = \termsof(\pi') \subseteq \subterms(T) \cup \subterms(E)$. Note that $\listsof(\pi) \subseteq \listsof(\pi') \cup \{[u]\}$, so the statement about lists is also true. 
        \item $\rnrule = \rnweak$: Let $\pi'$ be the immediate subproof. The result follows from IH and the fact that $\listsof(\pi) = \listsof(\pi') \cup \listsof(\alpha)$.
        \item $\rnrule = \rnlint$: All terms in the conclusion appear in some proper subproof, so the statement on terms follows by IH. None of the subproofs ends in $\rnlint$ or $\rnweak$ (and does not contain $\rncons$). Thus $\listsof(\pi') \subseteq \listsof(E) \cup \{[n] \mid n \in \subterms(T)\cup\subterms(E)]$, for every subproof $\pi'$. It follows that $\listsof(\pi) \subseteq \listsof(E \cup \{\alpha\}) \cup \{[n] \mid n \in \subterms(T) \cup \subterms(E \cup \{\alpha\})\}$. 
        
        \item $\rnrule = \rnsubst$: Let the major premise be $t\listmemb{\ell}$ and the minor premise be $\equals{t}{u}$. Both $t,u$ are from $\vars\cup\names$, and thus are in $\subterms(T) \cup \subterms(E)$. The result follows from IH.
        \item $\rnrule = \rnsays$: Let the major premise be $\beta$ and the minor premise be $\sk_{a}$. Since $T \DYderives \sk_{a}$, $\sk_{a} \in \subterms(T)$. And $\termsof(\pi) \subseteq \subterms(T) \cup \subterms(E) \cup \subterms(\beta) \cup \{\pk_{a}\} \subseteq \subterms(T) \cup \subterms(E \cup \{\alpha\})$.
        \qedhere
    \end{itemize}
\end{proof}

\end{document}

%% file: macros.tex
\newcommand{\defemph}{\textbf}

\newcommand{\dom}{{\sf dom}}
\newcommand{\rng}{{\sf rng}}
\newcommand{\restr}{\upharpoonright}
\newcommand{\prot}{\mathit{Pr}}

\newcommand{\recsend}[2]{{#1}\!\Rightarrow\!{#2}}
\renewcommand{\phi}{\varphi}
\renewcommand{\epsilon}{\varepsilon}

\newcommand{\knowfunc}{{\sf k}}

\newcommand{\knowupd}{{\it update}}

\newcommand{\names}{\mathscr{N}}

\newcommand{\Tterms}{\mathscr{T}}

\newcommand{\vars}{\mathscr{V}}

\newcommand{\qvar}{\mathscr{V}_{q}}
\newcommand{\ivar}{\mathscr{V}_{i}}
\newcommand{\ag}{\mathscr{A}}

\newcommand{\varsof}{{\sf vars}}

\newcommand{\freevars}{{\sf fv}}
\newcommand{\boundvars}{{\sf bv}}
\newcommand{\subterms}{{\sf st}}

\newcommand{\func}{{\sf f}}
\newcommand{\gunc}{{\sf g}}

\newcommand{\sk}{\mathit{sk}}
\newcommand{\pk}{\mathit{pk}}

\newcommand{\pair}{{\sf pair}}

\newcommand{\senc}{{\sf senc}}

\newcommand{\sign}{{\sf sign}}

\newcommand{\blind}{{\sf blind}}

\newcommand{\positions}[1]{\mathbb{P}(#1)}

\newcommand{\posof}[2]{\mathbb{P}_{#1}({#2})}
\newcommand{\subtermat}[2]{{#1}|_{#2}}
\newcommand{\replsubtermat}[3]{{#1}[{#3}]_{#2}}
\newcommand{\abstractable}{\mathbb{A}}

\newcommand{\nat}{\mathbb{N}}

\newcommand{\subformulas}{{\sf sf}}

\newcommand{\equals}[2]{{{#1}\bowtie{#2}}}
\newcommand{\listmemb}{\twoheadleftarrow}
\newcommand{\disj}{\vee}
\newcommand{\conj}{\wedge}

\newcommand{\says}{\ \mathit{says}\ }
\newcommand{\assertact}{{\sf assert}}
\newcommand{\denyact}{{\sf deny}}
\newcommand{\elg}{{\sf el}}

\newcommand{\publics}{{\sf pubs}}

\newcommand{\rnrule}{{\sf r}}

\newcommand{\rnpk}{{\sf pk}}

\newcommand{\rnpair}{{\sf pair}}
\newcommand{\rnsplit}{{\sf split}}

\newcommand{\rnsenc}{{\sf senc}}
\newcommand{\rnsdec}{{\sf sdec}}
\newcommand{\rnaenc}{{\sf aenc}}
\newcommand{\rnadec}{{\sf adec}}

\newcommand{\introd}{\sf i}
\newcommand{\elimin}{\sf e}

\newcommand{\rnax}{{\sf ax}}
\newcommand{\rneq}{{\sf eq}}
\newcommand{\rnsubst}{{\sf subst}}
\newcommand{\rncons}{{\sf cons}}
\newcommand{\rntrans}{{\sf trans}}
\newcommand{\rnsymm}{{\sf sym}}
\newcommand{\rnsays}{{\sf say}}

\newcommand{\rnweak}{{\sf wk}}
\newcommand{\rnproj}{{\sf proj}}

\newcommand{\rncut}{{\sf cut}}
\newcommand{\rnconji}{\wedge\introd}
\newcommand{\rnconje}{\wedge\elimin}

\newcommand{\rnexintro}{\exists\introd}
\newcommand{\rnexelim}{\exists\elimin}
\newcommand{\rnlint}{{\sf int}}
\newcommand{\rnlwk}{{\sf wk}}
\newcommand{\rnprom}{{\sf prom}}

\newcommand{\derives}{\vdash}
\newcommand{\DYderives}{\vdash_{\mathit{dy}}}
\newcommand{\DYnderives}{\nvdash_{\mathit{dy}}}
\newcommand{\assderives}{\vdash_{\mathit{a}}}

\newcommand{\eqderives}{\vdash_{\mathit{eq}}}
\newcommand{\measure}{\delta}

\newcommand{\newconseq}[3]{\mathscr{E}^{{#3}}_{{#1},{#2}}}

\newcommand{\dc}{\mathit{ker}}

\newcommand{\honterms}{\mathit{HT}}

\newcommand{\inteqs}{\mathit{IE}}
\newcommand{\intterms}{\mathit{IT}}

\newcommand{\intsub}{\sigma}
\newcommand{\vintsub}{\sigma^{\!*}}
\newcommand{\hwsub}{\theta}

\newcommand{\iwsub}{\mu}
\newcommand{\viwsub}{\mu^{\!*}}
\newcommand{\bigsub}{\omega}
\newcommand{\vbigsub}{\omega^{\!*}}
\newcommand{\vlambda}{\lambda^{\!*}}
\newcommand{\wh}{\widehat}

\newcommand{\dommu}{Z}
\newcommand{\expplus}[1]{{#1}^{*}}

\newcommand{\zap}[1]{\overline{\mspace{1mu}{#1}\mspace{1mu}}}

\newcommand{\dagsize}[1]{|\subterms({#1})|}

\newcommand{\fixedname}{{\sf m}}

\newcommand{\constst}{\mathscr{C}}
\newcommand{\stnonvars}{\mathscr{D}}
\newcommand{\hatst}{\wh{\constst}}

\newcommand{\rank}{\mathit{rank}}

\newcommand{\SF}{\subformulas}

\newcommand{\subterm}{\subterms}

\newcommand{\axiomsof}{{\sf axioms}}
\newcommand{\concof}{{\sf conc}}

\newcommand{\termsof}{{\sf terms}}
\newcommand{\listsof}{{\sf lists}}
\newcommand{\atomsof}{{\sf at}}
\newcommand{\eatomsof}{{\sf hat}}

%% file: arxiv.bbl
\begin{thebibliography}{10}

\bibitem{ABF17}
Mart{\'\i}n Abadi, Bruno Blanchet, and C{\'e}dric Fournet.
\newblock {The applied pi calculus: mobile values, new names, and secure communication}.
\newblock {\em Journal of the ACM}, 65(1):1:1--1:41, 2017.

\bibitem{AC06}
Mart{\'i}n Abadi and V{\'e}ronique Cortier.
\newblock {Deciding knowledge in security protocols under equational theories}.
\newblock {\em Theoretical Computer Science}, 367(1--2):2--32, 2006.

\bibitem{Adi08}
Ben Adida.
\newblock {Helios: web-based open-audit voting}.
\newblock In {\em 17th Conference on Security Symposium}, pages 335--348, 2008.

\bibitem{ALV03}
Roberto~M. Amadio, Denis Lugiez, and Vincent Vanack\'{e}re.
\newblock {On the symbolic reduction of processes with cryptographic functions}.
\newblock {\em Theoretical Computer Science}, 290(1):695--740, 2003.

\bibitem{ALRR17}
Myrto Arapinis, Jia Liu, Eike Ritter, and Mark Ryan.
\newblock {Stateful applied pi calculus: observational equivalence and labelled bisimilarity}.
\newblock {\em Journal of Logical and Algebraic Methods in Programming}, 89:95--149, 2017.

\bibitem{BHM08b}
Michael Backes, C\u{a}t\u{a}lin Hrit\c{c}u, and Matteo Maffei.
\newblock {Automated verification of remote electronic voting protocols in the applied pi-calculus}.
\newblock In {\em 21st IEEE Computer Security Foundations Symposium}, pages 195--209, 2008.

\bibitem{BMU08}
Michael Backes, Matteo Maffei, and Dominique Unruh.
\newblock {Zero-knowledge in the applied pi-calculus and automated verification of the Direct Anonymous Attestation protocol}.
\newblock In {\em 29th IEEE Symposium on Security and Privacy}, pages 202--215, 2008.

\bibitem{BRS10}
A.~Baskar, R.~Ramanujam, and S.~P. Suresh.
\newblock {A {\sc dexptime}-complete Dolev-Yao theory with distributive encryption}.
\newblock In {\em 35th International Symposium on Mathematical Foundations of Computer Science}, volume 6281 of {\em Lecture Notes in Computer Science}, pages 102--113, 2010.

\bibitem{Bau05}
Mathieu Baudet.
\newblock {Deciding security of protocols against off-line guessing attacks}.
\newblock In {\em 12th ACM Conference on Computer and Communications Security}, pages 16--25, 2005.

\bibitem{Bla01}
Bruno Blanchet.
\newblock {An efficient cryptographic protocol verifier based on Prolog rules}.
\newblock In {\em 14th IEEE Computer Security Foundations Workshop}, pages 82--96, 2001.

\bibitem{Bla16}
Bruno Blanchet.
\newblock {Modeling and verifying security protocols with the applied pi calculus and ProVerif}.
\newblock {\em Foundations and Trends in Privacy and Security}, 1(1):1--135, 2016.

\bibitem{BP05}
Bruno Blanchet and Andreas Podelski.
\newblock {Verification of cryptographic protocols: tagging enforces termination}.
\newblock {\em Theoretical Computer Science}, 333(1--2):67--90, 2005.

\bibitem{CKR18}
Vincent Cheval, Steve Kremer, and Itsaka Rakotonirina.
\newblock {The DEEPSEC prover}.
\newblock In {\em Computer Aided Verification}, volume 10982 of {\em Lecture Notes in Computer Science}, pages 28--36, 2018.

\bibitem{CKR20}
Vincent Cheval, Steve Kremer, and Itsaka Rakotonirina.
\newblock {The hitchhiker's guide to decidability and complexity of equivalence properties in security protocols}.
\newblock In {\em {Logic, Language, and Security: Essays Dedicated to Andre Scedrov on the Occasion of his 65th Birthday}}, volume 12300 of {\em Lecture Notes in Computer Science}, pages 127--145, 2020.

\bibitem{CKRT05}
Yannick Chevalier, Ralf K{\"u}sters, Micha{\"e}l Rusinowitch, and Mathieu Turuani.
\newblock {An NP decision procedure for protocol insecurity with XOR}.
\newblock {\em Theoretical Computer Science}, 338(1--3):247--274, 2005.

\bibitem{CS03}
Hubert Comon-Lundh and Vitaly Shmatikov.
\newblock {Intruder deductions, constraint solving and insecurity decisions in presence of exclusive or}.
\newblock In {\em 18th IEEE Symposium on Logic in Computer Science}, pages 271--280, 2003.

\bibitem{CDL06}
V{\'e}ronique Cortier, St{\'e}phanie Delaune, and Pascal Lafourcade.
\newblock {A survey of algebraic properties used in cryptographic protocols}.
\newblock {\em Journal of Computer Security}, 14(1):1--43, 2006.

\bibitem{CDS21}
V\'{e}ronique Cortier, St\'{e}phanie Delaune, and Vaishnavi Sundararajan.
\newblock {A decidable class of security protocols for both reachability and equivalence properties}.
\newblock {\em Journal of Automated Reasoning}, 65(4):479--520, 2021.

\bibitem{CK14}
V\'{e}ronique Cortier and Steve Kremer.
\newblock {Formal models and techniques for analyzing security protocols: a tutorial}.
\newblock {\em Foundations and Trends in Programming Languages}, 1(3):151--267, 2014.

\bibitem{CRZ05}
V{\'e}ronique Cortier, Micha{\"e}l Rusinowitch, and Eugen Z\u{a}linescu.
\newblock {A resolution strategy for verifying cryptographic protocols with CBC encryption and blind signatures}.
\newblock In {\em 7th ACM SIGPLAN International Conference on Principles and Practice of Declarative Programming}, pages 12--22, 2005.

\bibitem{Cre08}
Cas J.~F. Cremers.
\newblock {The Scyther tool: verification, falsification, and analysis of security protocols}.
\newblock In {\em 20th International Conference on Computer Aided Verification}, volume 5123 of {\em Lecture Notes in Computer Science}, pages 414--418, 2008.

\bibitem{DY83}
Danny Dolev and Andrew Yao.
\newblock {On the security of public-key protocols}.
\newblock {\em IEEE Transactions on Information Theory}, 29(2):198--208, 1983.

\bibitem{DLMS04}
Nancy Durgin, Patrick Lincoln, John Mitchell, and Andre Scedrov.
\newblock {Multiset rewriting and the complexity of bounded security protocols}.
\newblock {\em Journal of Computer Security}, 12(2):247--311, 2004.

\bibitem{FOO92}
Atsushi Fujioka, Tatsuaki Okamoto, and Kazuo Ohta.
\newblock {A practical secret voting scheme for large scale elections}.
\newblock In {\em Advances in Cryptology -- AUSCRYPT}, volume 718 of {\em Lecture Notes in Computer Science}, pages 244--251, 1992.

\bibitem{GMV22}
S{\'e}bastien Gondron, Sebastian M{\"o}dersheim, and Luca Vigan{\`o}.
\newblock {Privacy as reachability}.
\newblock In {\em 35th IEEE Computer Security Foundations Symposium}, pages 130--146, 2022.

\bibitem{GS08}
Jens Groth and Amit Sahai.
\newblock {Efficient non-interactive proof systems for bilinear groups}.
\newblock In {\em Advances in Cryptology -- EUROCRYPT}, volume 4965 of {\em Lecture Notes in Computer Science}, pages 415--432, 2008.

\bibitem{HT96}
Nevin Heintze and Doug Tygar.
\newblock {A model for secure protocols and their compositions}.
\newblock {\em IEEE Transactions on Software Engineering}, 22(1):16--30, 1996.

\bibitem{KK16}
Steve Kremer and Robert K{\"u}nnemann.
\newblock {Automated analysis of security protocols with global state}.
\newblock {\em Journal of Computer Security}, 24(5):583--616, 2016.

\bibitem{KR05}
Steve Kremer and Mark Ryan.
\newblock {Analysis of an electronic voting protocol in the applied pi calculus}.
\newblock In {\em Programming Languages and Systems -- ESOP 2005}, volume 3444 of {\em Lecture Notes in Computer Science}, pages 186--200, 2005.

\bibitem{LLT07}
Pascal Lafourcade, Denis Lugiez, and Ralf Treinen.
\newblock {Intruder deduction for the equational theory of abelian groups with distributive encryption}.
\newblock {\em Information and Computation}, 205(4):581--623, 2007.

\bibitem{MPR13}
Matteo Maffei, Kim Pecina, and Mathieu Reinert.
\newblock {Security and privacy by declarative design}.
\newblock In {\em 26th IEEE Computer Security Foundations Symposium}, pages 81--96, 2003.

\bibitem{McAll93}
David~A. McAllester.
\newblock {Automatic recognition of tractability in inference relations}.
\newblock {\em Journal of the ACM}, 40(2):284--303, 1993.

\bibitem{MSCB13}
Simon Meier, Benedikt Schmidt, Cas Cremers, and David Basin.
\newblock {The TAMARIN prover for the symbolic analysis of security protocols}.
\newblock In {\em 25th International Conference on Computer Aided Verification}, volume 8044 of {\em Lecture Notes in Computer Science}, pages 696--701, 2013.

\bibitem{MS01}
Jonathan~K. Millen and Vitaly Shmatikov.
\newblock {Constraint solving for bounded-process cryptographic protocol analysis}.
\newblock In {\em 8th ACM Conference on Computer and Communications Security}, pages 166--175, 2001.

\bibitem{RSS17}
R.~Ramanujam, Vaishnavi Sundararajan, and S.~P. Suresh.
\newblock {Existential assertions for voting protocols}.
\newblock In {\em Financial Cryptography and Data Security}, volume 10323 of {\em Lecture Notes in Computer Science}, pages 337--352, 2017.

\bibitem{RS05}
R.~Ramanujam and S.~P. Suresh.
\newblock {Decidability of context-explicit security protocols}.
\newblock {\em Journal of Computer Security}, 13(1):135--165, 2005.

\bibitem{RS06}
R.~Ramanujam and S.~P. Suresh.
\newblock {A (restricted) quantifier elimination for security protocols}.
\newblock {\em Theoretical Computer Science}, 367(1--2):228--256, 2006.

\bibitem{RT03}
Micha{\"e}l Rusinowitch and Mathieu Turuani.
\newblock {Protocol insecurity with finite number of sessions and composed keys is NP-complete}.
\newblock {\em Theoretical Computer Science}, 299(1--3):451--475, 2003.

\end{thebibliography}
